\documentclass[journal,10pt]{IEEEtran}

\usepackage{mathrsfs}

\usepackage[noadjust]{cite}
\usepackage{graphicx,color,overpic,psfrag}
\usepackage{amsmath, amssymb}
\usepackage{latexsym}
\usepackage{bm}
\usepackage{amssymb}
\usepackage{cases}
\usepackage{array}
\usepackage{fancyhdr}
\usepackage{stfloats}
\usepackage{setspace}
\usepackage{subfigure}
\usepackage{url}
\usepackage{algpseudocode}
\usepackage{algorithm}
\usepackage{blkarray}
\usepackage{booktabs}

\usepackage{multirow}
\usepackage{dsfont}
\usepackage{tabularx}

\usepackage{amsfonts}

\newtheorem{theorem}{Theorem}
\newtheorem{lemma}{Lemma}

\newcommand{\lmref}[1]{Lemma \ref{#1}}
\newcommand{\thref}[1]{Theorem \ref{#1}}

\newcommand{\figref}[1]{Fig. \ref{#1}}

\newcommand{\alref}[1]{Algorithm \ref{#1}}
\newcommand{\appref}[1]{Appendix \ref{#1}}

\newcommand{\Exp}{{\mathbb{E}}}
\newcommand{\expect}[1]{\Exp\left\{#1\right\}}

\newcommand{\tr}[1]{\mathrm{tr}\left(#1\right)}
\newcommand{\diag}[1]{\mathrm{diag}\left(#1\right)}
\newcommand{\logdet}[1]{\log\det\left(#1\right)}

\newcommand{\half}{\frac{1}{2}}

\newcommand{\ppd}[1]{ {\frac{ \partial }{ \partial {#1} }} }
\newcommand{\pppd}[2]{ {\frac{ \partial {#1} }{ \partial {#2} }} }

\newcommand{\argmax}[1]{\mathop{\arg\max}\limits_{#1}}

\newcommand{\maximize}[1]{\mathop{\mathrm{max}}\limits_{#1}}
\newcommand{\st}{\mathrm{subject}\quad\mathrm{to}}

\newcommand{\equaa}{\mathop{=}^{(\textrm{a})}}

\newcommand{\cR}{\mathcal{R}}

\newcommand{\bb}{\mathbf{b}}

\newcommand{\bh}{\mathbf{h}}

\newcommand{\bn}{\mathbf{n}}

\newcommand{\bx}{\mathbf{x}}
\newcommand{\by}{\mathbf{y}}

\newcommand{\bA}{\mathbf{A}}
\newcommand{\bB}{\mathbf{B}}

\newcommand{\bG}{\mathbf{G}}
\newcommand{\bH}{\mathbf{H}}
\newcommand{\bI}{\mathbf{I}}

\newcommand{\bK}{\mathbf{K}}

\newcommand{\bQ}{\mathbf{Q}}
\newcommand{\bR}{\mathbf{R}}

\newcommand{\bU}{\mathbf{U}}
\newcommand{\bV}{\mathbf{V}}

\newcommand{\bX}{\mathbf{X}}

\newcommand{\bbC}{\mathbb{C}}

\newcommand{\tildebx}{\tilde{\bx}}

\newcommand{\tildebK}{\tilde{\bK}}

\newcommand{\tildebH}{\tilde{\bH}}
\newcommand{\tildebR}{\tilde{\bR}}
\newcommand{\tildebQ}{\tilde{\bQ}}

\newcommand{\barbK}{\bar{\bK}}

\newcommand{\bzero}{\mathbf{0}}

\newcommand{\bLambda}{{\boldsymbol\Lambda}}
\newcommand{\bPhi}{{\boldsymbol\Phi}}

\newcommand{\bOmega}{{\boldsymbol\Omega}}
\newcommand{\bomega}{{\boldsymbol\omega}}
\newcommand{\bXi}{{\boldsymbol\Xi}}
\newcommand{\bPi}{{\boldsymbol\Pi}}

\newcommand{\bPsi}{{\boldsymbol\Psi}}

\newcommand{\bGamma}{{\boldsymbol\Gamma}}
\newcommand{\bDelta}{{\boldsymbol\Delta}}

\newcommand{\mm}{\mathrm{m}}
\newcommand{\eve}{\mathrm{eve}}
\newcommand{\tsec}{\mathrm{sec}}
\newcommand{\lb}{\mathrm{lb}}
\newcommand{\ub}{\mathrm{ub}}
\newcommand{\op}{\mathrm{op}}

\allowdisplaybreaks

\newcommand{\equadispwid}{0.35pt}

\begin{document}

\title{\huge Beam Domain Secure Transmission for Massive MIMO Communications}

\author{
Wenqian~Wu,~\IEEEmembership{Student~Member,~IEEE,}  Xiqi~Gao,~\IEEEmembership{Fellow,~IEEE,} Yongpeng~Wu,~\IEEEmembership{Senior~Member,~IEEE} and
Chengshan~Xiao,~\IEEEmembership{Fellow,~IEEE}

\thanks{

W. Wu and X. Q. Gao are with the National Mobile Communications Research Laboratory, Southeast University, Nanjing 210096, P. R. China (e-mail: wq\_wu@seu.edu.cn; xqgao@seu.edu.cn).

Y. Wu is with the Department of Electrical Engineering, Shanghai Jiao Tong University, Minhang 200240, China (e-mail: yongpeng.wu@sjtu.edu.cn).

C. Xiao is with the Department of Electrical and Computer Engineering, Lehigh University, Bethlehem, PA 18015, USA (e-mail: xiaoc@lehigh.edu).
}
}


\maketitle

\begin{abstract}
We investigate the optimality and power allocation algorithm of beam domain transmission for single-cell massive multiple-input multiple-output (MIMO) systems with a multi-antenna passive eavesdropper. Focusing on the secure massive MIMO downlink transmission with only statistical channel state information of legitimate users and the eavesdropper at base station, we introduce a lower bound on the achievable ergodic secrecy sum-rate, from which we derive the condition for eigenvectors of the optimal input covariance matrices. The result shows that beam domain transmission can achieve optimal performance in terms of secrecy sum-rate lower bound maximization. For the case of single-antenna legitimate users, we prove that it is optimal to allocate no power to the beams where the beam gains of the eavesdropper are stronger than those of legitimate users in order to maximize the secrecy sum-rate lower bound. Then, motivated by the concave-convex procedure and the large
dimension random matrix theory, we develop an efficient iterative and convergent algorithm to optimize power allocation in the beam domain. Numerical simulations demonstrate the tightness of the secrecy sum-rate lower bound and the near-optimal performance of the proposed iterative algorithm.
\end{abstract}


\section{Introduction}

For developing the next generation of wireless communication system, massive multiple-input multiple-output (MIMO) is considered as a promising technology to achieve larger gains in energy efficiency and spectral efficiency, and it has attracted significant interest from industry and academia \cite{Marzetta10Noncooperative,Larsson14Massive,Wang14Cellular}.
Massive MIMO transmission employs a large number of antennas at base station (BS) to serve a comparatively small number of users simultaneously.
In Marzetta's pioneering work \cite{Marzetta10Noncooperative}, which considered non-cooperative massive MIMO systems with single-antenna users and unlimited numbers of BS antennas, it was proven that the effects of fast fading and uncorrelated receiver noise vanish with the growth of the number of BS antennas, and the residual interference, known as pilot contamination, is induced by the reuse of the same pilot signals among adjacent cells.
Since the publication of \cite{Marzetta10Noncooperative}, various aspects of massive MIMO systems have been studied in recent year \cite{Jose11Pilot,Zhang13MAC,Hoydis13How,Ngo13Energy,Adhikary13JSDM,Sun15BDMA,You15Pilot,Lu16Deterministic,Chen16BDMA}.

Owing to the broadcast nature of the wireless medium, security is considered as a vital issue in wireless communication. Traditionally, key-based cryptographic techniques at the network layer were utilized to achieve communication security. However, these approaches are founded on certain assumptions for computational complexity, and are hence potentially vulnerable \cite{Mukherjee14Principles}. Recently, as a complement to key-based cryptographic techniques, physical layer security has attracted remarkable research interest, where information-theoretic security is investigated.
In Wyner's pioneering work \cite{Wyner75Noncooperative}, the wiretap channel, which consists of a transmitter, a legitimate user and an eavesdropper, was considered. It was revealed that if the the eavesdropper's channel is a degraded version of the legitimate user's channel, the transmitter can reliably send private message to the legitimate user, while the private message cannot be decoded by the eavesdropper. Then, more research has investigated physical layer security of multi-antenna systems \cite{Khisti10MISOME,Khisti10MIMOME,Oggier11Capacity,Geraci12Secrecy,Yang14Confidential,Wu12Linear,Cumanan14Secrecy}.
As shown in \cite{Goel08Guaranteeing,Zhou10Secure,Huang12Robust,Lin13Artificial}, if the transmitter only has the knowledge of the imperfect channel state information (CSI) of the eavesdropper, the security of the data transmission can be enhanced by transmitting artificial noise (AN) to disturb the decoding process at the eavesdropper. Moreover, for MIMO wiretap channels, the problem to determine the optimal input covariance matrix which maximizes the ergodic secrecy rate was studied in \cite{Li11Ergodic,Zhang11SWIPT}, where only statistical CSI of the legitimate user is required.

Recently, some works have been dedicated to investigating physical layer security in massive MIMO systems. J. Zhu \emph{et al.} investigated massive MIMO transmission with a passive eavesdropper \cite{Zhu14Secure,Zhu16Linear}, where only imperfect CSI of the legitimate users is available at the BS. An achievable secrecy rate and outage probability with AN generation and matched filter precoder at the BS was derived in \cite{Zhu14Secure}, assuming that the eavesdropper can perfectly cancel all interfering user signals. The work in \cite{Zhu14Secure} was further extended to the scenarios where the regular zero-forcing precoder and AN generation are employed at the BS in \cite{Zhu16Linear}. K. Guo \emph{et al.} proposed three secure transmission schemes for single-cell multi-user massive MIMO systems with distributed antennas in \cite{Guo16Security}. Moreover, Y. Wu \emph{et al.} studied secure transmission strategies for multi-cell massive MIMO systems, where a multi-antenna active eavesdropper can actively send the same pilot sequence as the users' to induce pilot contamination at the BS \cite{Wu16Active}. In \cite{Zhu14Secure,Zhu16Linear,Guo16Security,Wu16Active}, results were obtained by the assumption that instantaneous CSI of the legitimate users is available at BS and the legitimate users are equipped with single-antenna.

The availability of instantaneous CSI at the transmitter (CSIT) plays an important role in most existing transmission schemes for physical layer security.
Actually, in time-division duplex (TDD) communication systems, instantaneous CSIT is acquired via uplink training phase under the reciprocal channel assumption. However, for uplink/downlink RF hardware chains, this reciprocity is invalid at both BS and mobile transceivers \cite{Choi14FDD}.
Furthermore, the length of pilot signal is essentially limited by the channel coherence time. As for frequency-division duplex (FDD) communication systems, where the reciprocity of instantaneous channel is no longer valid, the required number of independent pilot symbols for CSIT acquisition and the CSIT feedback overhead scale with the number of BS antennas \cite{Jindal06MIMO}. Consequently, these practical limitations pose
severe challenges to acquire accurate instantaneous CSI of the legitimate channel at the BS in both TDD and FDD systems. Moreover, as mobility increases, the fluctuations of channel start to change more rapidly, and the round-trip delays of the CSI acquisition turn to be non-negligible with regard to the channel coherence time. For this case, transmitters may acquire outdated instantaneous CSI. Since the statistical channel parameters vary much more slowly than the instantaneous channel parameters, it is more reasonable to exploit the statistical CSIT for precoder design, when transmitters cannot easily acquire accurate instantaneous CSI. In addition, the uplink and downlink statistical CSIT are usually reciprocal in both TDD and FDD systems \cite{Barriac06Space}. Thus, the statistical CSIT can be acquired much more easily by utilizing this channel reciprocity, even if the terminals are equipped with multiple antennas.

C. Sun \emph{et al.} proposed a beam domain transmission for single-cell massive MIMO communications with only statistical CSIT \cite{Sun15BDMA}.
It was proven in \cite{Sun15BDMA} that beam domain transmission is optimal for a sum-rate upper bound maximization. However, the optimality of beam domain transmission for secure transmission with a multiple antenna eavesdropper was not provided in \cite{Sun15BDMA}. More importantly, the influence of the eavesdropper to the optimal transmit power allocation has not been studied in the literature to date.

In this paper, we investigate the secure transmission for single-cell downlink massive MIMO systems with a multiple antenna eavesdropper and multiple antenna legitimate users. Statistical CSIT under the jointly correlated MIMO channel model is considered. In massive MIMO systems, we note that as the number of BS antennas increases, the eigenmatrices of the channel transmit covariance matrices turn to be asymptotically unique and independent of mobile terminals \cite{You15Pilot}. For this case, a lower bound on the achievable ergodic secrecy sum-rate of the secure downlink transmission is derived.
Numerical simulations validate that the derived bound is tight in normal signal-to-noise ratio (SNR) ranges for practical applications such as long-term evolution (LTE), WiFi, and WiMax.
Based on this lower bound, the condition for eigenvectors of the optimal input covariance matrices which maximize the secrecy sum-rate lower bound is derived. It is proved that the beam domain transmission for single-cell downlink massive MIMO systems with statistical CSIT \cite{Sun15BDMA} is optimal for secure massive MIMO transmission in presence of a multiple antenna eavesdropper. As for the case of single-antenna legitimate users, we prove that allocating power to the beams where the beam gains of the eavesdropper are stronger than those of legitimate users will decrease the secrecy sum-rate lower bound. Then, we propose an efficient and fast iterative algorithm for power allocation in beam domain by using concave-convex procedure (CCCP) \cite{Yuille03CCCP} and large dimension random matrix theory. Numerical simulations demonstrate the near-optimal performance of the proposed fast-convergent iterative algorithm.

\textit{Notation}: Lower-case bold-face letters indicate vectors and upper-case bold-face letters indicate matrices. The matrix conjugate-transpose, conjugate, and transpose operations are denoted by superscripts $(\cdot)^H,\ (\cdot)^*$, and $(\cdot)^T$, respectively. $\bI_N$ and $\bzero_N$ denote $N\times N$ dimensional identity matrix and all-zero matrix, respectively. The subscript of the matrix is omitted for brevity, when the dimension of the matrix is clear. Also, $\tr{\cdot}, \expect{\cdot}$, and $\det(\cdot)$ denote matrix trace, ensemble expectation, and determinant operations, respectively. We use $\bA^{-1}$ and $[\cdot]_{mn}$ to represent the inverse of matrix $\bA$ and the $(m,n)$th element of matrix $\bA$, respectively.
$\bA\succeq \mathbf{0}$ indicates a positive semidefinite Hermitian matrix $\bA$. $\diag{\bb}$ indicates a diagonal matrix, whose main diagonal consists of the elements of vector $\bb$. $\odot$ stands for the Hadamard product, and $[x]^+$ represents $\max\{0,x\}$.

\section{System Model}\label{sec:system_model}

We consider secure downlink transmission in a single-cell massive MIMO system, cf. \figref{fig:sys_model}, consisting of an $M$-antenna BS, $K$ legitimate users, each with $N_{r}$ antennas, and a passive eavesdropper with $N_{e}$ antennas. The BS transmits private and independent messages to each legitimate user. All messages are required to be confidential to the eavesdropper. We note that neither the BS nor the users are assumed to know which user is eavesdropped, hence, we assume that any user may be potentially targeted by the eavesdropper.

\begin{figure}[!htbp]
\includegraphics[width=18em]{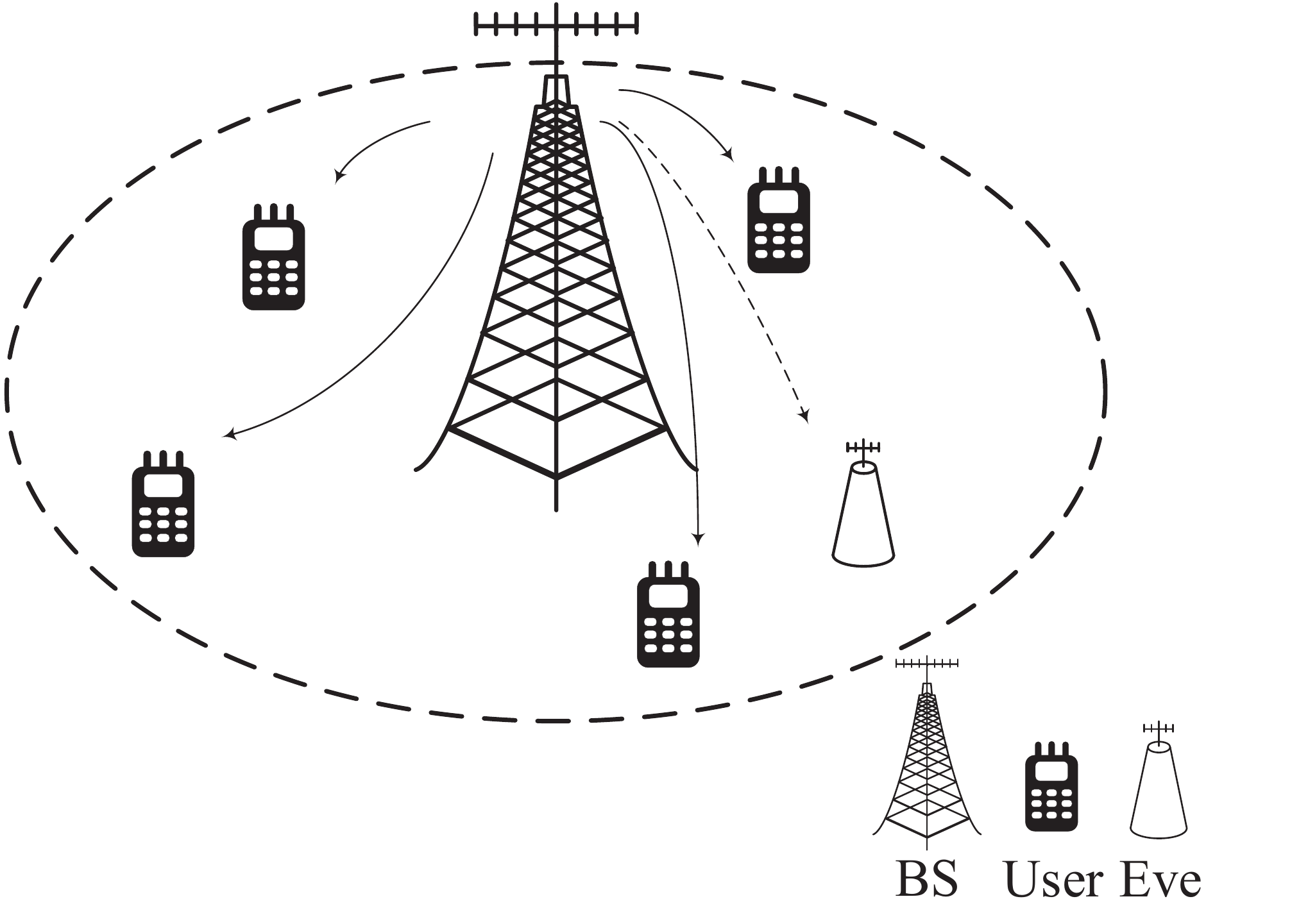}
\centering
\caption{Single-cell massive MIMO system with a multi-antenna passive eavesdropper.}\label{fig:sys_model}
\end{figure}

Let $\bH_{k}\in\bbC^{N_{r}\times M}$ and $\bH_{\eve}\in\bbC^{N_{e}\times M}$ denote the block fading channel matrices of the downlink channels from the BS to the $k$th user and the eavesdropper, respectively. The received signals at the $k$th user and at the eavesdropper are denoted by $\by_{k}\in\bbC^{N_{r}\times 1}$ and $\by_{\eve}\in\bbC^{N_{e}\times 1}$, respectively, and can be written as
\begin{align}
\by_{k}&=\bH_{k}\bx_{k}+\sum_{i\neq k}\bH_{k}\bx_{i}+\bn_{k}\\
\by_{\eve}&=\sum_{i}\bH_{\eve}\bx_{i}+\bn_{\eve}
\end{align}
where $\bx_{k}\in\bbC^{M\times 1}$ denotes the signal vector transmitted to the $k$th user which satisfies $\expect{\bx_{k}}=\bzero$, $\expect{\bx_{k}\bx_{k'}^{H}}=\bzero$ $(k\neq k')$, and $\expect{\bx_{k}\bx_{k}^{H}}=\bQ_{k}\in\bbC^{M\times M}$. $\bn_{k}\in\bbC^{N_{r}\times 1}$ and $\bn_{\eve}\in\bbC^{N_{e}\times 1}$ are zero-mean circularly symmetric complex Gaussian noise with covariance matrices $\bI_{N_r}$ and $\bI_{N_e}$, respectively. Here, without loss of generality, we consider a unit noise variance and assume that the BS has the power constraint
\begin{equation}\label{eq:pow_constraint}
  \sum_{k}\tr{\bQ_{k}}\leq P
\end{equation}
where $P\geq0$ depends on the BS power budget.

In this paper, the jointly correlated MIMO channel is adopted, which jointly model the correlation properties at the receiver and transmitter \cite{Gao09Eigenmode,Weichselberger06Channel}. Specifically, we can write
\begin{align}\label{eq:chan_model}
\bH_{k}&=\bU_{r,k}\bG_{k}\bV_{t,k}^{H}\\
\bH_{\eve}&=\bU_{r,\eve}\bG_{\eve}\bV_{t,\eve}^{H}
\end{align}
where $\bU_{r,k}\in\bbC^{N_{r}\times N_{r}}$,$\bU_{r,\eve}\in\bbC^{N_{e}\times N_{e}}$,$\bV_{t,k}\in\bbC^{M\times M}$, and $\bV_{t,\eve}\in\bbC^{M\times M}$ are deterministic unitary matrices, $\bG_{k}\in\bbC^{N_r \times M}$ and $\bG_{\eve}\in\bbC^{N_e \times M}$ are random matrices with zero-mean independent entries. In massive MIMO systems, as $M\rightarrow\infty$, unitary matrices $\bV_{t,k}$ and $\bV_{t,\eve}$ tend to be independent of mobile terminals and become a deterministic unitary matrix $\bV$ \cite{Sun15BDMA,You15Pilot}, which is only dependent on the topology of BS antenna array. Specially, if the BS is equipped with the uniform linear array (ULA), $\bV$ can be well approximated by the discrete Fourier transform (DFT) matrix \cite{Adhikary13JSDM}. Therefore, the beam domain channel matrices \cite{Sun15BDMA} can be defined as
\begin{align}
  \tildebH_{k}&=\bH_{k}\bV\label{eq:def_beam_chan1}\\
  \tildebH_{\eve}&=\bH_{\eve}\bV.\label{eq:def_beam_chan2}
\end{align}

Also, the eigenmode channel coupling matrices \cite{Gao09Eigenmode} can be defined as
\begin{align}\label{eq:def_coup_mat}
  \bOmega_{k}&=\expect{\bG_{k}\odot\bG_{k}^{*}}\\
  \bOmega_{\eve}&=\expect{\bG_{\eve}\odot\bG_{\eve}^{*}}.
\end{align}
The transmit correlation matrices can be expressed as
\begin{align}\label{eq:def_trans_corr}
  \bR_{k}&=\expect{\bH_{k}^{H}\bH_{k}}=\bV\tildebR_{k}\bV^{H}\\
  \bR_{\eve}&=\expect{\bH_{\eve}^{H}\bH_{\eve}}=\bV\tildebR_{\eve}\bV^{H}
\end{align}
where $\tildebR_{k}\in\bbC^{M\times M}$ and $\tildebR_{\eve}\in\bbC^{M\times M}$ are diagonal matrices with $[\tildebR_{k}]_{mm}=\sum_{n=1}^{N_r}[\bOmega_{k}]_{nm}$ and $[\tildebR_{\eve}]_{mm}=\sum_{n=1}^{N_e}[\bOmega_{\eve}]_{nm}$.

\textit{Remark 1}: The eigenmode channel coupling matrices $\bOmega_{k}$ of the legitimate users and $\bOmega_{\eve}$ of the eavesdropper are assumed to be perfectly known at the BS \cite{Li11Ergodic,Zhang11SWIPT,Wu16Active}.
From this point of view, we notice that it is reasonable to assume that the BS has the knowledge of the statistical CSI of the system terminals in massive MIMO systems. Thus, this assumption can be applied to the scenario where the BS aims to transmit private messages to some users while treating an idle user of the system as the eavesdropper.

In this paper, we assume that the legitimate users and the eavesdropper have instantaneous CSI of their corresponding channel matrices.  At each legitimate user, we treat the aggregate interference-plus-noise $\bn'_k=\sum_{i\neq k}\bH_{k}\bx_{i}+\bn_{k}$ as Gaussian noise with covariance matrix
\begin{equation}\label{eq:cov}
\bK_k=\bI+\sum_{i\neq k}\expect{\bH_{k}\bQ_{i}\bH_{k}^{H}}.
\end{equation}
Here, we assume the covariance matrix $\bK_k$ is known at the $k$th user. Besides, we make the pessimistic assumption that, at the the eavesdropper, signals of all legitimate users can be decoded and cancelled from the received signal $\by_{\eve}$ except the signal transmitted to the user of interest. \cite{Zhu14Secure,Wu16Active}. Since each user in the system has the risk of being eavesdropped, an achievable ergodic secrecy sum-rate can be expressed as \cite{Geraci12Secrecy}
\begin{equation}\label{eq:def_sumrate}
  R_{\tsec}=\sum_{k=1}^{K}[R_{k}-C_{k}^{\eve}]^{+}
\end{equation}
where
\begin{equation}\label{eq:def_user_rate}
  R_{k}=\expect{\logdet{\bI+\bK_{k}^{-1}\bH_{k}\bQ_{k}\bH_{k}^{H}}}
\end{equation}
and \cite{Wu16Active}
\begin{equation}\label{eq:def_eve_rate}
C_{k}^{\eve}=\expect{\logdet{\bI+\bH_{\eve}\bQ_{k}\bH_{\eve}^{H}}}
\end{equation}
where $R_{k}$ denotes an achievable ergodic rate between the BS and the $k$th user, and $C_{k}^{\eve}$ denotes the ergodic capacity between the BS and the eavesdropper, which seeks to decode the private messages intended for the $k$th user.

Notice that, in practical system, it is difficult to acquire instantaneous $\bH_{k}\bQ_{i}\bH_{k}^{H}$ ($i\neq k$) at the $k$th user in massive MIMO systems. Thus, we make an assumption that each legitimate user treats $\bn'_k$ as a Gaussian noise and the covariance matrix with expectation over $\bH_k$ is known at each user's side. With this assumption, the matrix $\bK_k$ defined in \eqref{eq:cov} is the covariance matrix of $\bn'_k$. Therefore, the ergodic rate defined in \eqref{eq:def_user_rate} is reasonable for practice.

In general, the secrecy sum-rate given by \eqref{eq:def_sumrate} is a non-concave function with respect to $(\bQ_{1},\cdots,\bQ_{K})$. Hence, it is difficult to determine the optimal input covariance matrices maximizing the exact secrecy sum-rate. Therefore, we introduce a lower bound on the achievable ergodic secrecy sum-rate, which is given by
\begin{equation}\label{eq:def_sec_lb}
  R_{\tsec,\lb}=\sum_{k=1}^{K}[R_{k}-C_{k,\ub}^{\eve}]^{+}
\end{equation}
where
\begin{align}
C_{k,\ub}^{\eve}&=\logdet{\bK_{\eve,k}}\label{eq:def_eve_ub}
\end{align}
with $\bK_{\eve,k}=\bI+\expect{\bH_{\eve}\bQ_{k}\bH_{\eve}^{H}}$. In \eqref{eq:def_sec_lb}, Jensen's inequality is adopted to obtain the upper bound of $C_{k}^{\eve}$ as in \eqref{eq:def_eve_ub} and consequently, a lower bound on the secrecy sum-rate.


Then, we design the secure transmission strategies by optimizing the secrecy sum-rate lower bound. Our main objective is to design the input covariance matrices $\bQ_{1},\cdots,\bQ_{K}$ maximizing \eqref{eq:def_sec_lb}, which can be formulated as the following optimization problem
\begin{align}\label{prob:def_prob1}
\left[\bQ_{1}^{\op},\cdots,\bQ_{K}^{\op}\right]=&\argmax{\bQ_{1},\cdots,\bQ_{K}}\sum_{k=1}^{K}\left(R_{k}-C_{k,\ub}^{\eve}\right)\notag\\
\st\quad &\tr{\sum_{k=1}^{K}\bQ_{k}}\leq P\notag\\
&\bQ_{k}\succeq \mathbf{0},\quad k=1,\cdots,K
\end{align}
where $(\bQ_{1}^{\op},\cdots,\bQ_{K}^{\op})$ is the optimal solution of the problem in \eqref{prob:def_prob1}. Because any negative term in the summation could increase to zero by setting the corresponding $\bQ_{k}=\bzero,\ k=1,\cdots,K$, the notation $[\cdot]^{+}$ is ignored when solving the problem in \eqref{prob:def_prob1}.

\section{Optimal Secure Transmission}\label{sec:optimal_transmission}

In this section, we first investigate the optimal transmission design based on the secrecy sum-rate lower bound in \eqref{eq:def_sec_lb}. Then, we consider a special case, where each legitimate user is equipped with single-antenna, and reveal the eavesdropper's impact on the optimal transmit power allocation.

\subsection{Optimality of Beam Domain Transmission}

Let $\bQ_{k}=\bPhi_{k}\bLambda_{k}\bPhi_{k}^{H}$, where $\bPhi_{k}$ is the eigenmatrix and $\bLambda_{k}$ is a diagonal matrix of the corresponding eigenvalues. In practice, $\bPhi_{k}$ and $\bLambda_{k}$ represent the directions in which signals are transmitted and the transmit power allocated onto each direction, respectively. For the beam domain transmission proposed in \cite{Sun15BDMA}, $\bPhi_{k}$ is set to be $\bV$, $k=1,...,K$. Next, we prove that this beam domain transmission structure is optimal for the secrecy sum-rate lower bound maximization problem in \eqref{prob:def_prob1}. We obtain the optimal input covariance matrix $\bQ_{k}^{\op}$ as follows.

\begin{theorem}\label{th:eigenmatrix}
The eigenmatrix of the optimal input covariance matrix of each legitimate user, maximizing the secrecy sum-rate lower bound as given by \eqref{eq:def_sec_lb}, is equal to that of the transmit correlation matrix of its own channel, i.e.,
\begin{equation}\label{eq:eigenmatrix_op}
\bPhi_{k}=\bV.
\end{equation}
\end{theorem}
\begin{proof}
See \appref{app:eigenmatrix_pf}.
\end{proof}

\textit{Remark 2}: From \eqref{eq:eigenmatrix_op}, the eigenmatrices of the input signals maximizing the secrecy sum-rate lower bound are given by the columns of $\bV$. This implies that beam domain secure transmission is optimal for the lower bound maximization. In addition, for a special case of the downlink multiuser transmission without secrecy constraint where $\bH_{\eve}=\bzero$, Theorem 1 reduces to the optimality condition derived in \cite{Sun15BDMA} in terms of maximizing an upper bound of the sum-rate.

Inspired by the result in Theorem 1, we now focus on the beam domain secure transmission. The received signals at the $k$th user and the eavesdropper can be rewritten as
\begin{align}\label{eq:def_receive}
\by_{k}&=\bH_{k}\bV\tildebx_{k}+\sum_{i\neq k}\bH_{k}\bV\tildebx_{i}+\bn_{k}\notag\\
&=\tildebH_{k}\tildebx_{k}+\sum_{i\neq k}\tildebH_{k}\tildebx_{i}+\bn_{k}\\
\by_{\eve}&=\sum_{i}\bH_{\eve}\bV\tildebx_{i}+\bn_{\eve}\notag\\
&=\sum_{i}\tildebH_{\eve}\tildebx_{i}+\bn_{\eve}
\end{align}
where $\tildebx_{k}=\bV^{H}\bx_{k}$ is the beam domain transmitted signals whose covariance matrix is $\bLambda_{k}$, and the power constraint can be rewritten as $\sum_{k}\tr{\bLambda_{k}}\leq P$.

With $\bPhi_{k}=\bV$, $\tildebK_{k}=\bI+\sum_{i\neq k}\expect{\tildebH_{k}\bLambda_{i}\tildebH_{k}^{H}}$, and $\tildebK_{\eve,k}=\bI+\expect{\tildebH_{\eve}\bLambda_{k}\tildebH_{\eve}^{H}}$, the secrecy sum-rate lower bound in \eqref{eq:def_sec_lb} can be rewritten as
\begin{align}\label{eq:beam_lb}
R_{\tsec,\lb}=&\sum_{k=1}^{K}\Bigg[\expect{\log\det\bigg(\bI+\tildebK_{k}^{-1}\tildebH_{k}\bLambda_{k}\tildebH_{k}^{H}\bigg)}\notag\\
&-\logdet{\tildebK_{\eve,k}}\Bigg]^{+}.
\end{align}
Theorem 1 provides an optimal transmit direction to maximize the secrecy sum-rate lower bound based on the statistical CSI. Next, we discuss the optimal transmit power allocation for maximizing the lower bound.

\subsection{Property of Optimal Power Allocation}

Now we focus on the design of eigenvalues, i.e., the transmit power on each beam. The power allocation problem can be formulated as

\begin{align}\label{prob:def_prob2}
\left[\bLambda_{1}^{\op},\!\cdots\!,\!\bLambda_{K}^{\op}\right]\!=\!&\argmax{\bLambda_{1},\cdots,\bLambda_{K}}\sum_{k=1}^{K}\!\Bigg(\!\expect{\!\log\det\bigg(\!\tildebK_{k}\!+\!\tildebH_{k}\bLambda_{k}\tildebH_{k}^{H}\!\bigg)\!}\notag\\
&-\logdet{\tildebK_{k}}-\logdet{\tildebK_{\eve,k}}\Bigg)\notag\\
\st\quad &\sum_{k}\tr{\bLambda_{k}}\leq P\notag\\
&\bLambda_{k}\succeq \mathbf{0},\quad k=1,\cdots,K
\end{align}
where $(\bLambda_{1}^{\op},\cdots,\bLambda_{K}^{\op})$ is the solution of the above optimization problem.
The original problem in \eqref{prob:def_prob1} is equivalent to the power allocation problem in \eqref{prob:def_prob2}. For obtaining insight for transmit signal design, a lemma is introduced as follows.

\begin{lemma}\label{lm:ineq}
For a positive random variable $x$, it holds that
\begin{equation}
\expect{\frac{x}{a+bx}}\leq\expect{\frac{\bar{x}}{a+bx}}
\end{equation}
where $\bar{x}=\expect{x}$, $a>0$, and $b>0$. The equality holds if and only if $x=\bar{x}$ with probability one for every $x$.
\end{lemma}
\begin{proof}
See \appref{app:ineq_pf}.
\end{proof}

Unlike the power allocation problem without secrecy constraint in \cite{Sun15BDMA}, for secure massive MIMO transmission, the eavesdropper has impact on the optimal transmit power allocation. In particular, for secure transmission among single-antenna users, we can obtain the following theorem with the help of Lemma 1.

\begin{theorem}\label{th:eigenvalue2}
When each legitimate user is equipped with single-antenna ($N_{r}=1$), the solution of power allocation problem in \eqref{prob:def_prob2} holds
\begin{equation}\label{eq:eigenvalue2}
[\bLambda_{k}^{\op}]_{mm}=0,\quad \text{if }[\tildebR_{\eve}]_{mm}\ge[\tildebR_{k}]_{mm}.
\end{equation}
\end{theorem}
\begin{proof}
See \appref{app:eigenvalue2_pf}.
\end{proof}

Theorem 2 reveals that the optimal transmitted beam sets of different single-antenna users should not contain the beams where the beam gains of the eavesdropper is stronger than that of legitimate user.
In other words, for the optimal power allocation, power only should be allocated among the beams where legitimate users have stronger beam gains than those of the eavesdropper in order to maximize the secrecy sum-rate lower bound.

In general, the optimal power allocation in \eqref{prob:def_prob2} does not have closed-form expression. In the next section, we will provide an efficient algorithm to solve this problem.

\section{Iterative Algorithm for Power Allocation}\label{sec:iterative_algorithm}

In this section, an efficient and fast iterative algorithm for power allocation in \eqref{prob:def_prob2} is developed. First, we transform the original non-convex problem in \eqref{prob:def_prob2} to a series of convex programs by CCCP algorithm, which can be used to find the local optimum of the problem in \eqref{prob:def_prob2}. Next, in order to reduce the computation complexity, the deterministic equivalent of the secrecy sum-rate lower bound is utilized to develop an efficient power allocation algorithm.

Obtaining the optimal power allocation in \eqref{prob:def_prob2} is difficult because of the difference of convex (d.c.) objective functions. To tackle this problem, we introduce CCCP algorithm to solve this power allocation problem. Define
\begin{align}
\cR_{k,1}(\bLambda_{1},\cdots,\bLambda_{K})&=\expect{\logdet{\tildebK_{k}+\tildebH_{k}\bLambda_{k}\tildebH_{k}^{H}}}\notag\\
&\equaa\expect{\logdet{\barbK_{k}+\bG_{k}\bLambda_{k}\bG_{k}^{H}}}\label{eq:cR}\\
\cR_{k,2}(\bLambda_{1},\cdots,\bLambda_{K})&=\logdet{\tildebK_{k}}+\logdet{\tildebK_{\eve,k}}\notag\\
&\equaa\logdet{\barbK_{k}}+\logdet{\barbK_{\eve,k}}
\end{align}
where $\barbK_{k}=\bI+\sum_{i \neq k}\expect{\bG_k\bLambda_{i}\bG_k^H}$ and
$\barbK_{\eve,k}=\bI+\expect{\bG_{\eve}\bLambda_{k}\bG_{\eve}^H}$.
Here, (a) follows from the fact that $\logdet{\bI+\bA\bB}=\logdet{\bI+\bB\bA}$.
Then, we exploit the CCCP algorithm, which transforms the problem in \eqref{prob:def_prob2} into a series of convex programs as follows,
\begin{align}\label{prob:def_prob4}
&\left[\bLambda_{1}^{(i+1)},\cdots,\bLambda_{K}^{(i+1)}\right]=\argmax{\bLambda_{1},\cdots,\bLambda_{K}}\sum_{k}
\cR_{k,1}(\bLambda_{1},\cdots,\bLambda_{K})\notag\\
&\quad\quad-\sum_{k}\tr{\left(\ppd{\bLambda_{k}}\sum_{l}\cR_{l,2}(\bLambda_{1}^{(i)},\cdots,\bLambda_{K}^{(i)})\right)^{T}\bLambda_{k}}\notag\\
&\qquad\qquad\st\quad\sum_{k}\tr{\bLambda_{k}}\leq P\notag\\
&\qquad\qquad\qquad\qquad\qquad\bLambda_{k}\succeq \mathbf{0},\quad k=1,\cdots,K.
\end{align}

For the problem in \eqref{prob:def_prob4}, we note that the CCCP algorithm is a majorize-minimize algorithm. In CCCP algorithm, the concave part $\cR_{k,2}$ is linearized around the solution of current iteration such that the objective function is concave on $(\bLambda_{1},\cdots,\bLambda_{K})$. Subsequently, the non-convex optimization problem in \eqref{prob:def_prob2} is tackled as a series of concave problem in \eqref{prob:def_prob4}.
However, without closed-form expression, evaluating $\cR_{k,1}(\bLambda_{1},\cdots,\bLambda_{K})$ can be computationally cumbersome. To evade Monte-Carlo averaging over the legitimate user channels, we calculate the deterministic equivalent instead of $\cR_{k,1}(\bLambda_{1},\cdots,\bLambda_{K})$ by large dimension random matrix theory.
Following the approach of \cite{Lu16Deterministic}, the closed-form expression of the deterministic equivalent of \eqref{eq:cR} can be calculated as
\begin{align}\label{eq:DE_R1}
&\overline{\cR}_{k,1}(\bLambda_{1},\cdots,\bLambda_{K})=\logdet{\bI+\bGamma_{k}\bLambda_{k}}\notag\\
&\qquad+\logdet{\tilde{\bGamma}_{k}+\barbK_{k}}-\tr{\bI-\tilde{\bPhi}_{k}^{-1}}
\end{align}
where $\bGamma_{k}\in\bbC^{M\times M}$ and $\tilde{\bGamma}_{k}\in\bbC^{N_{r}\times N_{r}}$ are given by
\begin{align}
\bGamma_{k}&=\eta_{k}\left(\tilde{\bPhi}_{k}^{-1}\barbK_{k}^{-1}\right)\label{eq:cal_gam}\\
\tilde{\bGamma}_{k}&=\tilde{\eta}_{k}\left(\bPhi_{k}^{-1}\bLambda_{k}\right)\label{eq:cal_tgam}.
\end{align}
$\tilde{\bPhi}_{k}\in\bbC^{N_{r}\times N_{r}}$ and $\bPhi_{k}\in\bbC^{M\times M}$ are obtain by the iterative equations
\begin{align}
\tilde{\bPhi}_{k}&=\bI+\tilde{\eta}_{k}\left(\bPhi_{k}^{-1}\bLambda_{k}\right)\barbK_{k}^{-1}\label{eq:cal_tphi}\\
\bPhi_{k}&=\bI+\eta_{k}\left(\tilde{\bPhi}_{k}^{-1}\barbK_{k}^{-1}\right)\bLambda_{k}\label{eq:cal_phi}.
\end{align}
Moreover, $\tilde{\eta}_{k}(\bX)\in\bbC^{N_{r}\times N_{r}}$ and $\eta_{k}(\tilde{\bX})\in\bbC^{M\times M}$ are diagonal matrices, whose diagonal entries are given by
\begin{align}
[\tilde{\eta}_{k}(\bX)]_{nn}&=\sum_{m=1}^{M}[\bOmega_{k}]_{nm}[\bX]_{mm}\\
[\eta_{k}(\tilde{\bX})]_{mm}&=\sum_{n=1}^{N_{r}}[\bOmega_{k}]_{nm}[\tilde{\bX}]_{nn}.
\end{align}
Thus, the deterministic equivalent of the secrecy sum-rate lower bound can be expressed as
\begin{align}\label{eq:deter_lb}
\!\overline{R}_{\tsec,\lb}\!=\!\sum_{k=1}^{K}\!\bigg[\overline{\cR}_{k,1}(\bLambda_{1},\!\cdots\!,\bLambda_{K})\!-\!\cR_{k,2}(\bLambda_{1},\!\cdots\!,\bLambda_{K}) \!\bigg]^{+}\!.\!
\end{align}

Note that the deterministic equivalent $\overline{\cR}_{k,1}(\bLambda_{1},\cdots,\bLambda_{K})$ depends on the correlation matrices $\tilde{\eta}_{k}(\bX)$ and $\eta_{k}(\tilde{\bX})$ \cite{Lu16Deterministic}, which can be calculated effectively. From \cite{Dumont10Capacity,Dupuy11Capacity}, we can find that the deterministic equivalent $\overline{\cR}_{k,1}(\bLambda_{1},\cdots,\bLambda_{K})$ is strictly concave on $(\bLambda_{1},\cdots,\bLambda_{K})$. Thus, we turn to consider the following series of convex programs instead of \eqref{prob:def_prob4},
\begin{align}\label{prob:def_prob5}
&\left[\bLambda_{1}^{(i+1)},\cdots,\bLambda_{K}^{(i+1)}\right]=\argmax{\bLambda_{1},\cdots,\bLambda_{K}}\sum_{k}
\overline{\cR}_{k,1}(\bLambda_{1},\cdots,\bLambda_{K})\notag\\
&\quad\quad-\sum_{k}\tr{\left(\ppd{\bLambda_{k}}\sum_{l}\cR_{l,2}(\bLambda_{1}^{(i)},\cdots,\bLambda_{K}^{(i)})\right)^{T}\bLambda_{k}}\notag\\
&\qquad\qquad\st\quad\sum_{k}\tr{\bLambda_{k}}\leq P\notag\\
&\qquad\qquad\qquad\qquad\qquad\bLambda_{k}\succeq \mathbf{0},\quad k=1,\cdots,K.
\end{align}

\textit{Remark 3}: The sequence generated by \eqref{prob:def_prob5} has proven to be convergent and approximately optimal in \cite{Chen16BDMA}. We note that the CCCP algorithm is an effective way to solve the d.c. problem, where the solution of \eqref{prob:def_prob4} is a local optimum of the d.c. problem in \eqref{prob:def_prob2}. Meanwhile, with the purpose of computation complexity reduction, we calculate the deterministic equivalent of $\cR_{k,1}(\bLambda_{1},\cdots,\bLambda_{K})$. As will be shown in Section \ref{sec:Numerical_Results}, the results of deterministic equivalent are nearly identical to those of the Monte-Carlo simulation.
Although the solution of problem in \eqref{prob:def_prob5} is an approximate solution of problem in \eqref{prob:def_prob4}, calculating deterministic equivalent can significantly reduce the computation complexity by avoiding Monte-Carlo averaging over the channels.

Define diagonal matrices and sets as
\begin{align}
&\bDelta_{k}^{(i)}\!\!=\!\!\sum_{l\neq k}\!\sum_{j=1}^{N_{r}}\!\frac{\check{\bR}_{l,j}}{1\!+\!\tr{\!\bLambda_{\backslash l}^{(i)}\check{\bR}_{l,j}\!}}\!+\!\sum_{j=1}^{N_{e}}\!\frac{\check{\bR}_{\eve,j}}{1\!+\!\tr{\!\bLambda_{k}^{(i)}\check{\bR}_{\eve,j}\!}}\label{eq:cal_delta}\!\\
&\mathcal{S}_{k,m,l}\!=\!\bigg\{\big(l',m'\big)\bigg|l'\!\neq\!l, \ \big(l',m'\big)\!\neq\!\big(k,m\big), \ l'\!\in\!\big\{1,\!\cdots\!,K\big\}, \notag\\
&\qquad\qquad\quad m'\in\big\{1,\cdots,M\big\}\bigg\}
\end{align}
where $\bLambda_{\backslash l}^{(i)}=\sum_{l'\neq l}\bLambda_{l'}^{(i)}$, $\check{\bR}_{l,j}=\diag{\bomega_{l,j}}$, and $\check{\bR}_{\eve,j}=\diag{\bomega_{\eve,j}}$. Here, $\bomega_{l,j}$ and $\bomega_{\eve,j}$ are the $j$th row of $\bOmega_{l}$ and $\bOmega_{\eve}$, respectively.
Based on \eqref{eq:cal_gam}-\eqref{eq:cal_phi}, we can define $\bGamma_{k}^{(i+1)}$ and $\tilde{\bGamma}_{k}^{(i+1)}$ by $\bLambda_{1}^{(i+1)},\cdots,\bLambda_{K}^{(i+1)}$, while
letting the $m$th diagonal entries of $\bGamma_{k}^{(i+1)},\ \tilde{\bGamma}_{k}^{(i+1)},\ \check{\bR}_{l,j},\ \bDelta_{k}^{(i)}$ and $\bLambda_{k}^{(i)}$ be $\gamma_{k,m}^{(i+1)},\ \tilde{\gamma}_{k,m}^{(i+1)},\ \check{r}_{l,m,j},\ \delta_{k,m}^{(i)}$ and $\lambda_{k,m}^{(i)}$, respectively.
Utilizing a similar procedure in \cite{Lu16Deterministic}, we can obtain the theorem as follows.
\begin{theorem}\label{th:deterministic}
The solution of the iterative problem in \eqref{prob:def_prob5} is equivalent to that of the problem in \eqref{prob:def_prob6}
\begin{align}\label{prob:def_prob6}
&\left[\bLambda_{1}^{(i+1)},\cdots,\bLambda_{K}^{(i+1)}\right]=\argmax{\bLambda_{1},\cdots,\bLambda_{K}}\sum_{k}
\bigg(\log\det\Big(\bI+\bGamma_{k}\bLambda_{k}\Big)\notag\\
&\qquad+\logdet{\tilde{\bGamma}_{k}+\barbK_{k}}-\tr{\bDelta_{k}^{(i)}\bLambda_{k}}\bigg)\notag\\
&\qquad\qquad\st\quad\sum_{k}\tr{\bLambda_{k}}\leq P\notag\\
&\qquad\qquad\qquad\qquad\qquad\bLambda_{k}\succeq \mathbf{0},\quad k=1,\cdots,K.
\end{align}
The $m$th element $\lambda_{k,m}^{(i+1)}$ of $\bLambda_{k}^{(i+1)}$ satisfies \eqref{eq:det_solution}, which is given at the top of the next page.
\begin{figure*}[t]
\begin{align}\label{eq:det_solution}
\begin{cases}
\frac{\gamma_{k,m}^{(i+1)}}{1+\gamma_{k,m}^{(i+1)}\lambda_{k,m}^{(i+1)}}+\sum\limits_{l \neq k}^{K}\sum\limits_{j=1}^{N_{r}}\frac{\check{r}_{l,m,j}}{\tilde{\gamma}_{l,j}^{(i+1)}+\tr{\check{\bR}_{l,j}\bLambda_{\backslash l}^{(i+1)}}}=\delta_{k,m}^{(i)}+\mu^{(i+1)},\quad\quad&\mu^{(i+1)}<\nu_{k,m}^{(i+1)}-\delta_{k,m}^{(i)}\\
\lambda_{k,m}^{(i+1)}=0,\quad\quad&\mu^{(i+1)}\geq\nu_{k,m}^{(i+1)}-\delta_{k,m}^{(i)}
\end{cases}
\end{align}
\hrulefill
\end{figure*}
In \eqref{eq:det_solution}, the auxiliary variable $\nu_{k,m}^{(i+1)}$ is given by
\begin{equation}
\nu_{k,m}^{(i+1)}=\gamma_{k,m}^{(i+1)}\!+\!\sum_{l \neq k}^{K}\sum_{j=1}^{N_{r}}\frac{\check{r}_{l,m,j}}{\tilde{\gamma}_{l,j}^{(i+1)}\!+\!\sum\limits_{\substack{(l',m')\\\in\mathcal{S}_{k,m,l}}}\!
\check{r}_{l,m',j}\lambda_{l',m'}^{(i+1)}}
\end{equation}
and $\mu^{(i+1)}$ is chosen to satisfy the KKT conditions $\mu^{(i+1)}\left(\tr{\sum_{k}\bLambda_{k}^{(i+1)}}-P\right)=0$ and $\mu^{(i+1)}\geq0$.
\end{theorem}
\begin{proof}
See \appref{app:deterministic_pf}.
\end{proof}

\textit{Remark 4}: The solution \eqref{eq:det_solution} has the similar structure to the classical water-filling solution. The distinction lies in multiple terminals which causes a summation in the equation. Thus, generally, it is difficult to obtain the solution \eqref{eq:det_solution} and numerical approaches are required. For the case of $K=1$ (single-user), if the power constraint $\tr{\sum_{k}\bLambda_{k}^{(i+1)}}=P$ is considered, the solution is given by $\lambda_{k,m}^{(i+1)}=\big[(\delta_{k,m}^{(i)}+\mu^{(i+1)})^{-1}-(\gamma_{k,m}^{(i+1)})^{-1}\big]^+$, where $\mu^{(i+1)}$ is chosen to satisfy the constraint $\tr{\sum_{k}\bLambda_{k}^{(i+1)}}=P$.

For the sake of convenience, we define
\begin{align}
&\tilde{C}^{(i)}(\bLambda_{1},\cdots,\bLambda_{K})=\sum_{k}\bigg(\!\!\logdet{\bI+\bGamma_{k}^{(i)}\bLambda_{k}}\notag\\
&\qquad+\logdet{\tilde{\bGamma}_{k}^{(i)}+\barbK_{k}}-\tr{\bDelta_{k}^{(i)}\bLambda_{k}}\!\!\bigg)\label{eq:tildeC}\\
&\rho_{k,m}^{(i)}(x_{k,m})=\frac{\gamma_{k,m}^{(i)}}{1+\gamma_{k,m}^{(i)}x_{k,m}}-\delta_{k,m}^{(i)}-\mu\notag\\
&\quad+\sum\limits_{l \neq k}^{K}\sum\limits_{j=1}^{N_{r}}\frac{\check{r}_{l,m,j}}
{\tilde{\gamma}_{l,j}^{(i)}+\check{r}_{l,m,j}x_{k,m}+\!\sum\limits_{\substack{(l',m')\\\in\mathcal{S}_{k,m,l}}}\!\check{r}_{l,m',j}x_{l',m'}^{(t)}}\label{eq:rho1}\\
&\rho_{k,m}'^{(i)}(x_{k,m})=-\frac{(\gamma_{k,m}^{(i)})^{2}}{(1+\gamma_{k,m}^{(i)}x_{k,m})^{2}}\notag\\
&\quad-\sum\limits_{l \neq k}^{K}\!\sum\limits_{j=1}^{N_{r}}\!\frac{\check{r}_{l,m,j}^{2}}{(\tilde{\gamma}_{l,j}^{(i)}\!+\!\check{r}_{l,m,j}x_{k,m}
\!+\!\sum\limits_{\substack{(l',m')\\\in\mathcal{S}_{k,m,l}}}\check{r}_{l,m',j}x_{l',m'}^{(t)})^{2}}.\label{eq:rho2}
\end{align}
Then, to acquire the solution of the iterative problem in \eqref{prob:def_prob6}, a deterministic equivalent based iterative algorithm is summarized as \alref{alg:CCCP}, and specifically, to obtain $\bLambda_{1}^{(i+1)},\cdots,\bLambda_{K}^{(i+1)}$ in each iteration in \alref{alg:CCCP}, the iterative water-filling algorithm (IWFA) is utilized, which is described in \alref{alg:IWFA}.

\begin{algorithm}[!htbp]
\caption{Deterministic equivalent based iterative algorithm.}
\label{alg:CCCP}
\begin{algorithmic}[1]
\State Initialize $\{\bLambda_{1}^{(0)},\cdots,\bLambda_{K}^{(0)}\}$, thresholds $\xi_{1}$, $\xi_{2}$, and set iteration $i=0$.
\State \textbf{repeat}
\State \quad Initialize $u=0$, and $\tilde{\bPhi}_{k}^{(u)}$.
\State \quad \textbf{repeat}
\State \quad \quad Calculate $\tilde{\bPhi}_{k}^{(u+1)}$ and $\bPhi_{k}^{(u+1)}$ by \eqref{eq:cal_tphi} and \eqref{eq:cal_phi}.
\State \quad \quad set $u=u+1$.
\State \quad \textbf{until} $|\tilde{\bPhi}_{k}^{(u+1)}-\tilde{\bPhi}_{k}^{(u)}|\leq \xi_{1}$.
\State \quad Calculate $\bGamma_{k}^{(i)}$ and $\tilde{\bGamma}_{k}^{(i)}$ by \eqref{eq:cal_gam} and \eqref{eq:cal_tgam}, $k=1,\cdots\!,K$.
\State \quad Calculate $\bDelta_{k}^{(i)}$ based on \eqref{eq:cal_delta}, $k=1,\cdots,K$.
\State \quad Utilize \alref{alg:IWFA} to update $\bLambda_{k}^{(i+1)}, k=1,\cdots,K.$
\State \quad Set $i=i+1$, and calculate $\overline{R}_{\tsec,\lb}^{(i)}$ by \eqref{eq:deter_lb}.
\State \textbf{until} $|\overline{R}_{\tsec,\lb}^{(i)}-\overline{R}_{\tsec,\lb}^{(i-1)}|\leq \xi_{2}$.
\end{algorithmic}
\end{algorithm}

\begin{algorithm}[!htbp]
\caption{Iterative water-filling algorithm.}
\label{alg:IWFA}
\begin{algorithmic}[1]
\State Initialize diagonal matrices $\bX_{k}^{(0)}=\bLambda_{k}^{(i)}, k=1,\cdots,K$, $\tilde{C}^{(i)}(\bX_{1}^{(0)},\cdots,\bX_{K}^{(0)})$, and set iteration $t=0$.
\State \textbf{repeat}
\State \quad Initialize $u=0$ and $\mu^{(u)}=0$.
\State \quad \textbf{repeat}
\State \quad \quad \textbf{for} $k=1$ to $K$ \textbf{do}
\State \quad \quad \quad \textbf{for} $m=1$ to $M$ \textbf{do}
\State \quad \quad \quad \quad \textbf{repeat}
\State \quad \quad \quad \quad \quad Calculate $\rho_{k,m}^{(i)}(x_{k,m}^{(w)})$ and $\rho_{k,m}'^{(i)}(x_{k,m}^{(w)})$ by \eqref{eq:rho1} and \eqref{eq:rho2}.
\State \quad \quad \quad \quad \quad Update $x_{k,m}$ as
\Statex
\vspace{-2mm}
\begin{equation*}
x_{k,m}^{(w+1)}=x_{k,m}^{(w)}-\rho_{k,m}^{(i)}(x_{k,m}^{(w)})/\rho_{k,m}'^{(i)}(x_{k,m}^{(w)}).
\vspace{1mm}
\end{equation*}
\State \quad \quad \quad \quad \quad Set $w=w+1$.
\State \quad \quad \quad \quad \textbf{until} $|x_{k,m}^{(w+1)}-x_{k,m}^{(w)}|\leq\xi_{3}$
\State \quad \quad \quad \textbf{end for}
\State \quad \quad \textbf{end for}
\State \quad \quad Update $\bar{x}_{k,m}\!=\![x_{k,m}^{(w+1)}]^{+}$, and Calculate
\Statex
\vspace{-2mm}
\begin{equation*}
p_{tot}=\sum_{k,m}\bar{x}_{k,m}, \ k=1,\cdots,K; \ m=1,\cdots,M.
\vspace{1mm}
\end{equation*}
\State \quad \quad \textbf{if} $\mu^{(0)}=0$ and $p_{tot}\leq P$ \textbf{then}
\State \quad \quad \quad \textbf{go to} Step 20.
\State \quad \quad \textbf{end if}
\State \quad \quad Update $\mu^{(u+1)}=\mu^{(u)}+\Delta\mu$, where $\Delta\mu$ can be obtain as
\Statex
\vspace{-2mm}
\begin{equation*}
\!\Delta\mu=\min_{k,m}\bigg\{\Big|\rho_{k,m}^{(i)}(\bar{x}_{k,m}+(P-p_{tot})/M)-\rho_{k,m}^{(i)}(\bar{x}_{k,m})\Big|\bigg\}
\vspace{1mm}
\end{equation*}
\Statex and set $u=u+1$.
\State \quad \textbf{until} $|P-p_{tot}|\leq \xi_{4}$.
\State \quad Update $x_{k,m}^{(t+1)}=\frac{1}{KM}\bar{x}_{k,m}+\frac{KM-1}{KM}x_{k,m}^{(t)},\ k=1,\cdots,K;\ m=1,\cdots,M.$
\State \quad Set $t=t\!+\!1$, and calculate $\tilde{C}^{(i)}(\bX_{1}^{(t)},\!\cdots\!,\bX_{K}^{(t)})$ by \eqref{eq:tildeC}.
\State \textbf{until}
\Statex
\vspace{-2mm}
\begin{equation*}
\Big|\tilde{C}^{(i)}(\bX_{1}^{(t)},\!\cdots\!,\bX_{K}^{(t)})-\tilde{C}^{(i)}(\bX_{1}^{(t-1)},\!\cdots\!,\bX_{K}^{(t-1)})\Big|\leq \xi_{5}.
\vspace{1mm}
\end{equation*}
\State Update $\bLambda_{k}^{(i+1)}=\bX_{k}^{t},\ k=1,\cdots,K.$
\end{algorithmic}
\end{algorithm}

For the convergence of \alref{alg:CCCP} and \alref{alg:IWFA}, owing to the utilization of \alref{alg:IWFA} in each iteration of \alref{alg:CCCP}, we first prove the convergence of the proposed \alref{alg:IWFA}. We define $\tilde{C}_{max}^{(i)}$ as the objective function value corresponding to the solution  $\bLambda_{1}^{(i+1)},\cdots,\bLambda_{K}^{(i+1)}$, which is the maximum of the objective function of problem \eqref{prob:def_prob6} in the $i$th iteration with the power constraint $\sum_{k}\tr{\bLambda_{k}}\leq P$. Then, we can obtain the following theorem for \alref{alg:IWFA}.

\begin{theorem}\label{th:alg}
IWFA is a convergent algorithm, where the sequence $\{\tilde{C}^{(i)}(\bX_{1}^{(t)},\cdots,\bX_{K}^{(t)})\}_{t=0}^{\infty}$ generated in \alref{alg:IWFA} converges to $\tilde{C}_{max}^{(i)}$.
\end{theorem}
\begin{proof}
See \appref{app:alg_pf}.
\end{proof}

As for \alref{alg:CCCP}, by using the monotonically increasing property of CCCP, the sequence $\{\bLambda_{1}^{(i)},\cdots,\bLambda_{K}^{(i)}\}^{\infty}_{i=0}$ generated by \alref{alg:CCCP} can be easily proven to be convergent. Thus, the proof of the convergence of \alref{alg:CCCP} is omitted.

\textit{Remark 5}: For the sake of the summation of fraction functions, Newton's method \cite{Cormen09Newton} is applied to obtain approximate roots in Step 9 of the IWFA. In addition, in order to make \alref{alg:IWFA} converge fast, we can update $\bX_{k}^{(t+1)}$ with the results
$\bar{\bX}_{k}$ obtained in Step 20 of IWFA if $\tilde{C}^{(i)}(\bar{\bX}_{1},\cdots,\bar{\bX}_{K})$ is increased, whereas the relation \eqref{eq:IWFA_conv} suggests to update with $\left(\frac{1}{KM}\bar{\bX}_{k}+\frac{KM-1}{KM}\bX_{k}^{(t)}\right)$. To guarantee the convergence, we still use $\left(\frac{1}{KM}\bar{\bX}_{k}+\frac{KM-1}{KM}\bX_{k}^{(t)}\right)$ to update when the result $\tilde{C}^{(i)}(\bar{\bX}_{1},\cdots,\bar{\bX}_{K})$ is not increasing in each iteration.
Also, to make \alref{alg:CCCP} converge fast, we can initialize $(\bLambda_{1}^{(0)},\cdots,\bLambda_{K}^{(0)})$ by allocating equal power among the non-overlapping beams where the beam gains of legitimate users are much stronger than those of the eavesdropper. Note that simulations in the following section have demonstrated that both \alref{alg:CCCP} and \alref{alg:IWFA} can converge within a few iterations.

Finally, we discuss the computational complexities of \alref{alg:CCCP} and \alref{alg:IWFA}. For each iteration in \alref{alg:CCCP}, we utilize \alref{alg:IWFA} to obtain $\bLambda_{1}^{(i+1)},\cdots,\bLambda_{K}^{(i+1)}$. Since the iterations of diagonal matrices $\tilde{\bPhi}_{k}^{(u+1)}$ and $\bPhi_{k}^{(u+1)}$ can quickly converge and the complexity of calculating $\tilde{\bPhi}_{k}^{(u+1)}$ and $\bPhi_{k}^{(u+1)}$ is relatively low, the major complexity for one iteration of \alref{alg:CCCP} is constituted by the complexity of \alref{alg:IWFA}. As for \alref{alg:IWFA}, we adopt Newton's method \cite{Cormen09Newton} to solve the equation, consisting of the summation of fraction functions. If the precision is set to $g$ digits, the convergence of Newton's method will require $\log g$ iterations \cite{Cormen09Newton}. For the outer iteration of $\tilde{C}^{(i)}(\bX_{1}^{(t)},\cdots,\bX_{K}^{(t)})$ in \alref{alg:IWFA}, it will converge within a few iterations as the simulation illustration in the following section. Thus, the complexity of one iteration in \alref{alg:CCCP} will approximate to the complexity of inner iteration of \alref{alg:IWFA}, which is $O(KM\log g+KM)$. Then, the whole complexity of \alref{alg:CCCP} can be approximated by $O(LKM\log g+LKM)$, where $L$ is the iteration times of \alref{alg:CCCP}. As shown in the next section, $L$ is usually small.

\section{Numerical Results}\label{sec:Numerical_Results}

In this section, numerical results are provided to evaluate the secrecy performance of the beam domain secure transmission. Since the jointly correlated channel is a good approximation for WINNER II channel model, we utilize the WINNER II channel model to generate $\bH_k$ and $\bH_{\eve}$. Note that the WINNER II channel model is a geometry-based stochastic channel model. In our simulations, the suburban scenario under the non-line-of-sight (NLOS) condition is considered. Both BS and terminals are equipped with ULAs. Neither shadow fading nor path loss is considered. We consider the legitimate users and the eavesdropper are uniformly distributed within the cell. Here, in Step 1 in \alref{alg:CCCP}, when initialize $(\bLambda_{1}^{(0)},\cdots,\bLambda_{K}^{(0)})$, we choose $16$ strongest beams by comparing $[\tildebR_{k}]_m^m-[\tildebR_{\eve}]_m^m$, $k=1,\cdots,K$, $m=1,\cdots,M$. Then, we allocate equal power among these beams.

\begin{figure}[!htbp]
\includegraphics[width=25em]{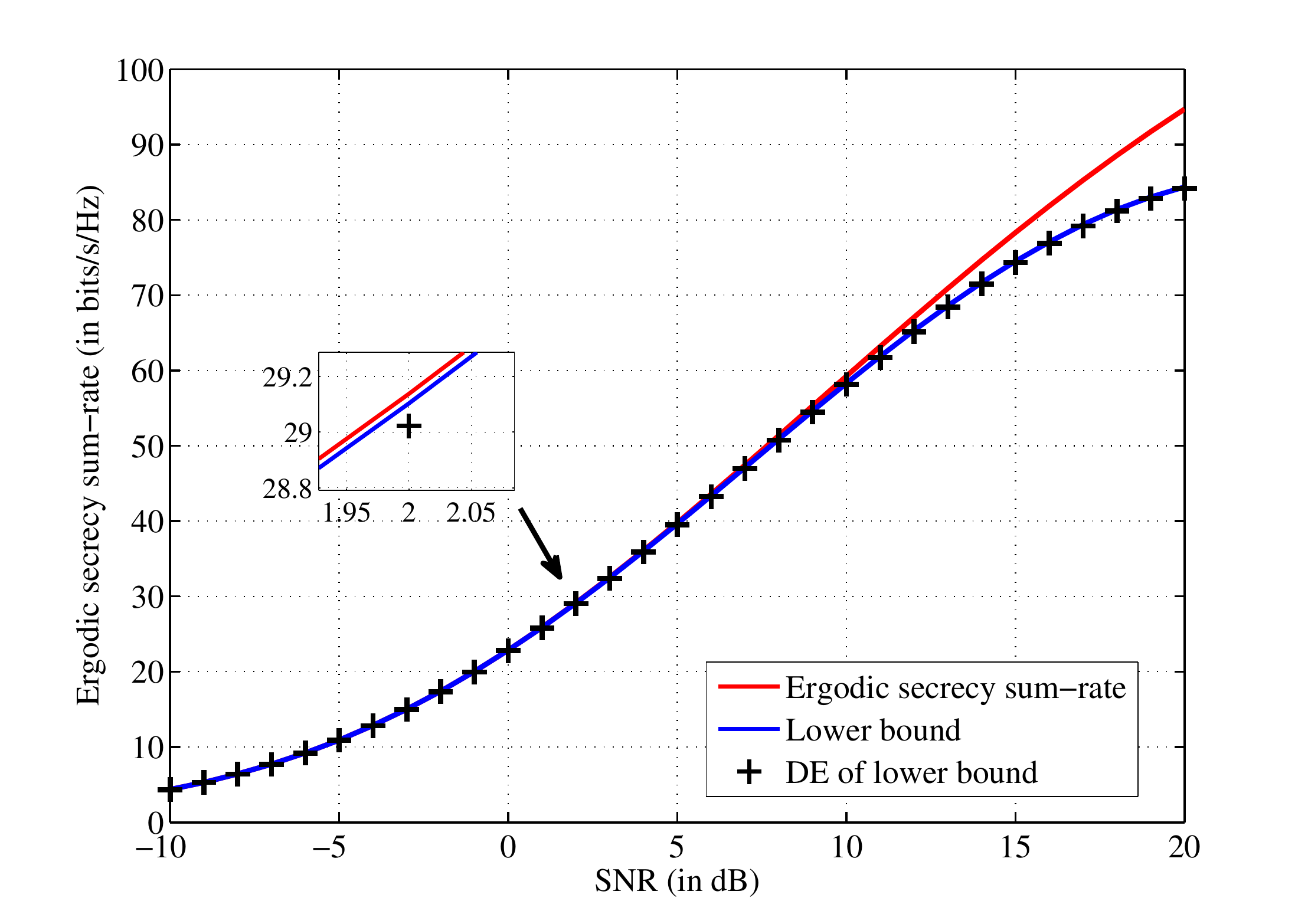}
\centering
\caption{Comparison of the ergodic secrecy sum-rate and the lower bound. Results are shown versus the SNRs of the WINNER II channel with $M=128, N_r=4, N_e=4$, and $K=8$. The deterministic equivalent of the lower bound is also depicted.}\label{fig:compare_lb}
\end{figure}

\begin{figure}[!htbp]
\includegraphics[width=25em]{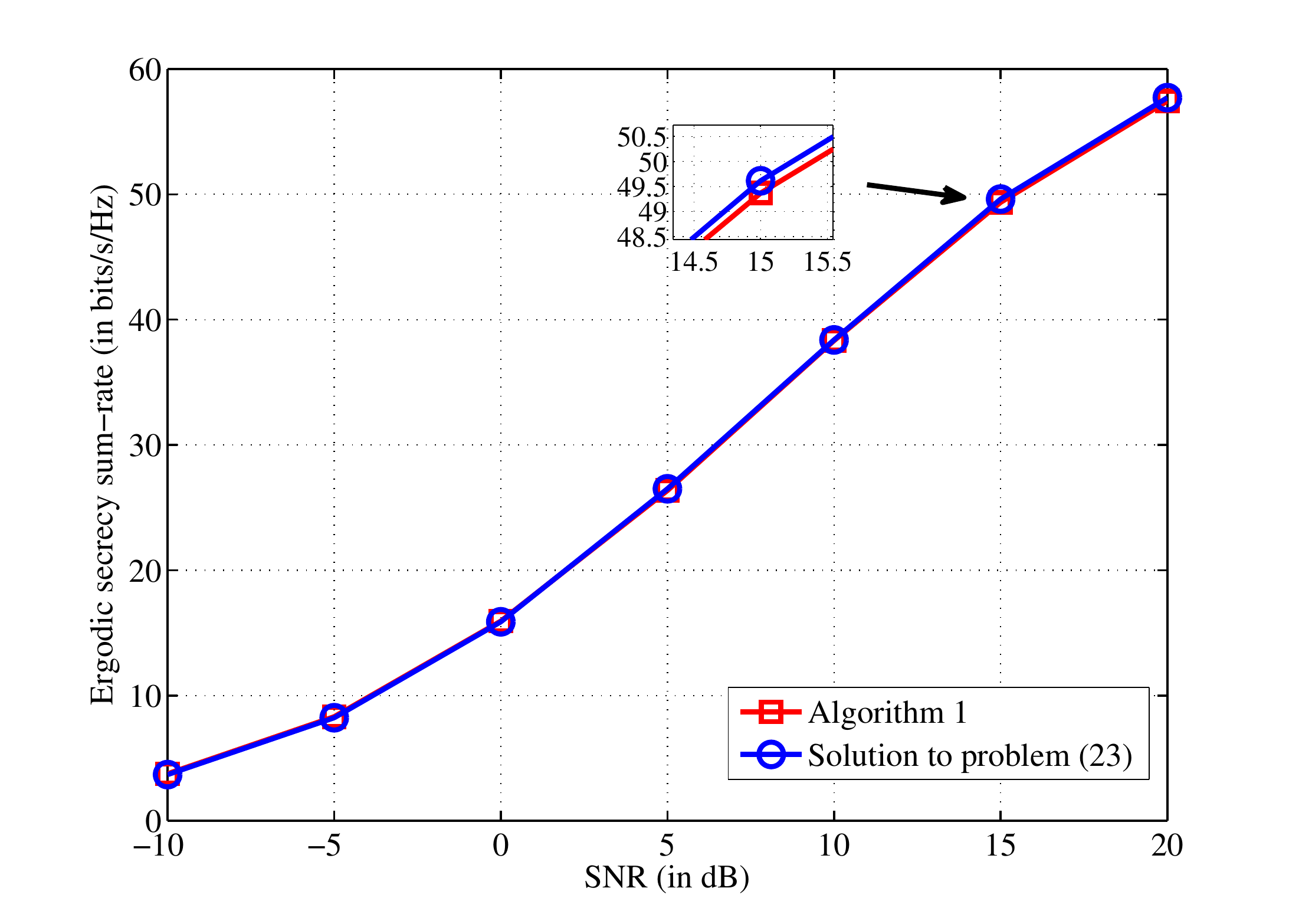}
\centering
\caption{Exact ergodic secrecy sum-rate for the proposed iterative power allocation algorithm in \alref{alg:CCCP} and the solution of power allocation problem \eqref{prob:def_prob2}. Results are shown versus the SNRs of the WINNER II channel with
$M=64,N_r=4,N_e=4$, and $K = 4$.}\label{fig:compare_opt}
\end{figure}

\figref{fig:compare_lb} compares the ergodic secrecy sum-rate \eqref{eq:def_sumrate} with its lower bound \eqref{eq:def_sec_lb}, which are evaluated by Monte-Carlo simulations. Here, we consider $M=128, N_r=4, N_e=4$, and $K=8$. We simulate from SNR $=-10$ dB to SNR $=20$ dB. These are normal SNR ranges for practical applications such as LTE, WiFi, and WiMax \cite{Vu07MIMO}. \figref{fig:compare_lb} illustrates that the lower bound is quite tight in the considered SNR ranges. Moreover, we compare the Monte-Carlo simulation results and the deterministic equivalent of the lower bound given in \eqref{eq:deter_lb}. We find that the deterministic equivalent results are nearly identical to the Monte-Carlo simulation results, which is also depicted in \figref{fig:compare_lb}.

\begin{figure}[!htbp]
\includegraphics[width=25em]{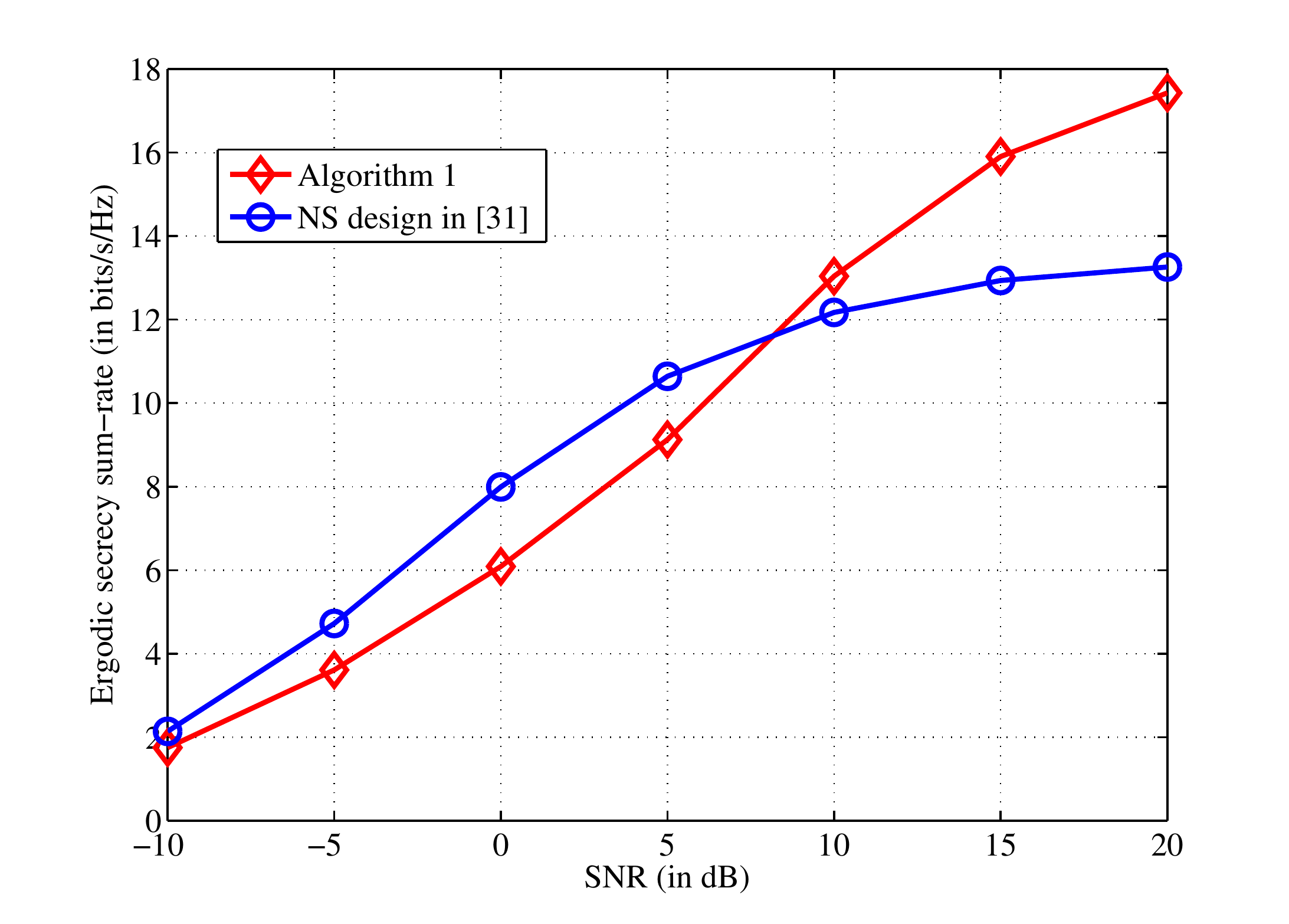}
\centering
\caption{Exact ergodic secrecy sum-rate for the proposed iterative power allocation algorithm in \alref{alg:CCCP}, the NS design in \cite{Wu16Active}. Here, instantaneous CSIT is utilized in \cite{Wu16Active} while statistical CSIT is utilized in \alref{alg:CCCP}. Results are shown versus the SNRs of the WINNER II channel with $M=128, N_r=1, N_e=4, K=12$. }\label{fig:alg_performance}
\end{figure}

\figref{fig:compare_opt} compares the performance of \alref{alg:CCCP} against that of the optimal power allocation for secrecy lower bound maximization. Here, the optimal power allocation is obtained by solving problem \eqref{prob:def_prob2} by utilizing the optimization toolbox in Matlab. For the sake of computational complexity, we set $M=64, N_r=4, N_e=4$, and $K=4$. As shown in \figref{fig:compare_opt}, the performance of these methods are nearly identical, which indicates the near-optimal performance of \alref{alg:CCCP}. Furthermore, the computational complexity of \alref{alg:CCCP} is much lower than the complexity of solving problem \eqref{prob:def_prob2} by optimization toolbox in Matlab.

\figref{fig:alg_performance} compares the performance of \alref{alg:CCCP} against that of the null space (NS) design in \cite{Wu16Active} with $M=128, N_r=1, N_e=4$, and $K=12$. In the NS design, orthogonal pilot sequences and matched filter precoding are used for $K$ single-antenna legitimate users, and artificial noise is not used. Moreover, the NS design in \cite{Wu16Active} refers to transmitting in NS of the transmit correlation matrix of the eavesdropper.
Here, we do not consider pilot contamination attack and the minimum mean square estimate is adopted to obtain the instantaneous CSIT of legitimate users. Note that instantaneous CSIT is required in the NS design while \alref{alg:CCCP} only need statistical CSIT.
Our simulation results show that \alref{alg:CCCP} can achieve a much higher secrecy sum-rate at high SNR, and the secrecy sum-rate performance of \alref{alg:CCCP} is only slightly lower than that of the NS design at low SNR.
This is because the interference from other users has much influence on secrecy sum-rate at high SNR, where power allocation can be beneficial while matched filter precoding is unhelpful.
Furthermore, as with \cite{Wu16Active}, transmitting in the NS of the transmit correlation matrix of the eavesdropper is not a optimal scheme.
For the case of low SNR, since power allocation can bring little secrecy sum-rate gain when the total transmitted power is comparable to the noise power, the performance of NS design is better, where instantaneous CSIT is utilized.

\begin{figure}[!htbp]
\includegraphics[width=25em]{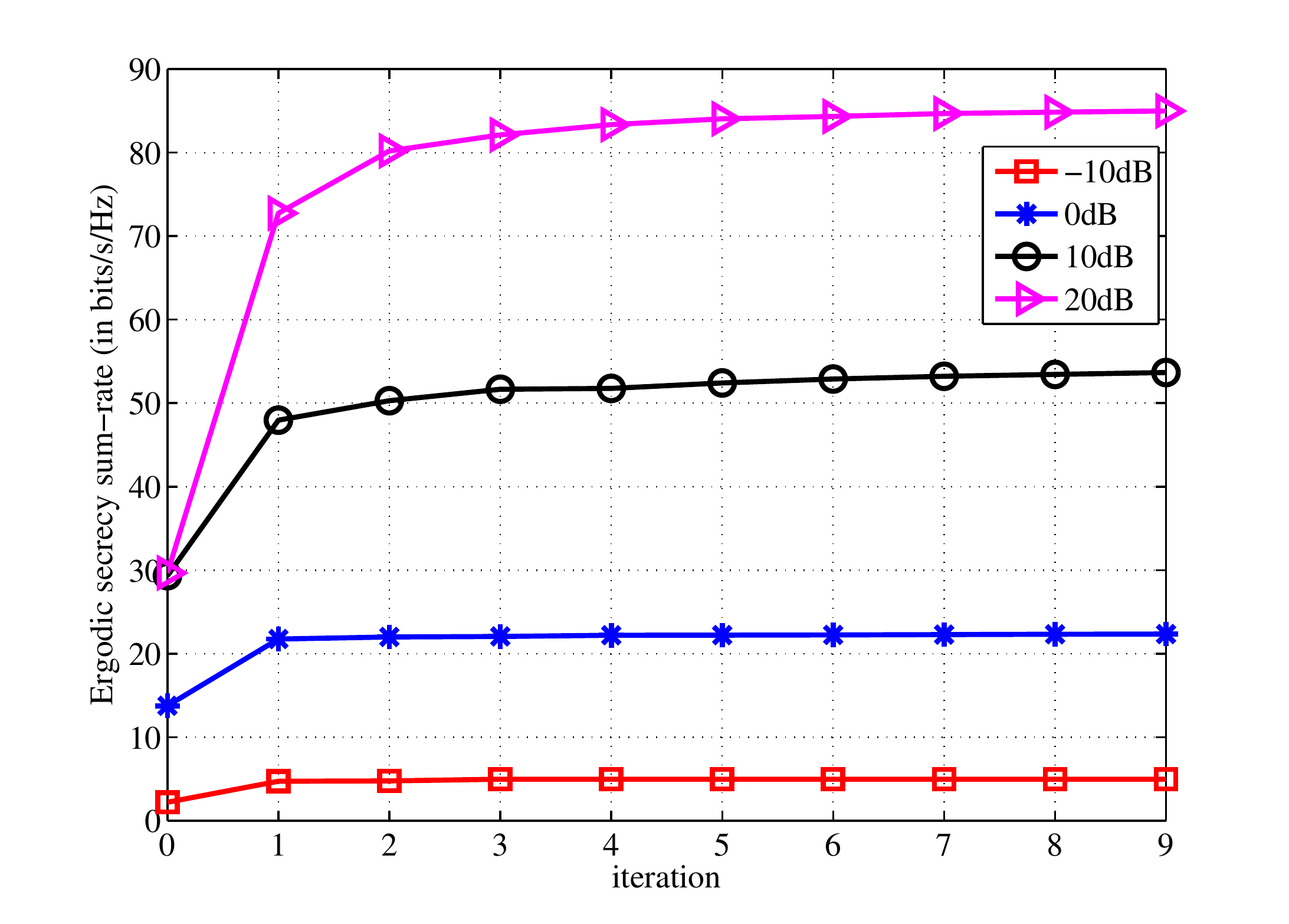}
\centering
\caption{Convergence of \alref{alg:CCCP}. Results are shown versus the iteration times for different SNRs with $M=128, N_r=4, N_e=4$, and $K=8$.}\label{fig:alg1_convergence}
\end{figure}

\begin{figure}[!htbp]
\includegraphics[width=25em]{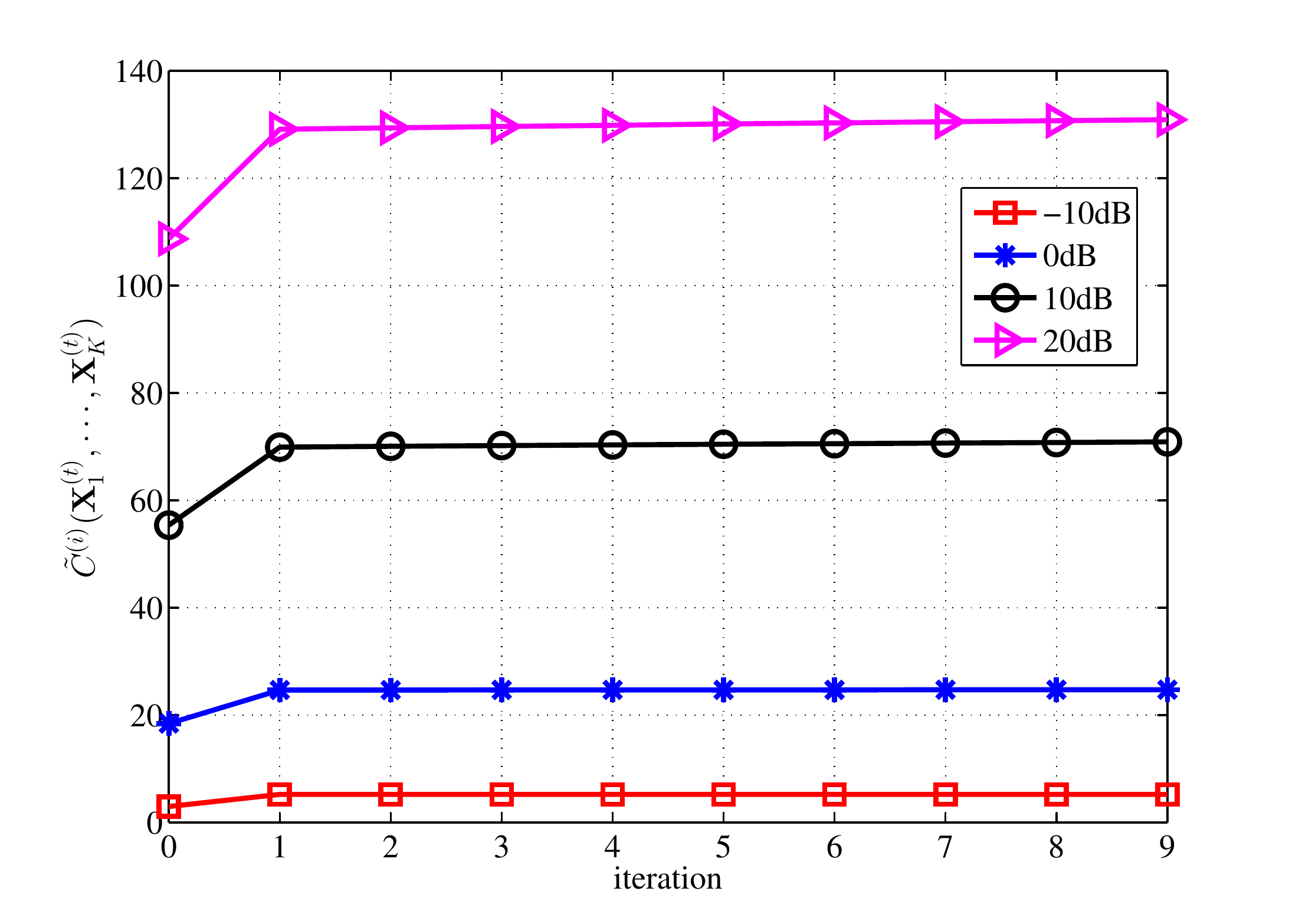}
\centering
\caption{Convergence of \alref{alg:IWFA}. Results are shown versus the iteration times for different SNRs with $M=128, N_r=4, N_e=4$, and $K=8$.}\label{fig:alg2_convergence}
\end{figure}

In order to further illustrate the performance of the algorithms, the convergence of the proposed \alref{alg:CCCP} and \alref{alg:IWFA} are shown in \figref{fig:alg1_convergence} and \figref{fig:alg2_convergence}, respectively. Here, we set the SNR as $-10$ dB, $0$ dB, $10$ dB and $20$ dB, respectively. For \alref{alg:CCCP}, the optimal performance can be approached after the first iteration at low SNR, and the near-optimal performance can still be achieved with in a few iterations at high SNR. From this result, we note that the convergence of \alref{alg:CCCP} turns to be slower, as the SNR increases. It is because that the total transmitted power is comparable to the noise power at low SNR, which indicates little secrecy sum-rate gain acquired from power allocation. However, for the case of high SNR, since the interference from other users has more effect on secrecy sum-rate, power allocation can be helpful. Therefore, more iterations are required. As for \alref{alg:IWFA}, we can find that the proposed \alref{alg:IWFA} converges after only a few iterations, especially, the proposed \alref{alg:IWFA} could achieve near-optimal performance after the first iteration in all cases.

\section{Conclusion}\label{sec:Conclusion}

In this work, we have investigated downlink single-cell massive MIMO transmission with a multi-antenna passive eavesdropper, where only statistical CSI of legitimate users and the eavesdropper is available at the BS. As the number of BS antennas goes to infinity, the eigenmatrices of the channel transmit covariance matrices turn to be identical and independent of mobile terminals. Based on this characteristic, we obtained the condition for eigenmatrix of the optimal input covariance matrix to maximize secrecy sum-rate lower bound. Our analysis showed that the eigenmatrix of the optimal input covariance matrices, maximizing the secrecy sum-rate lower bound, is equal to that of the transmit correlation matrix. Thus, beam domain transmission, satisfying this condition, can achieve optimal performance with respect to secrecy sum-rate lower bound maximization. Also, for the case of single-antenna legitimate users, we revealed that it is optimal to allocate no power to the beams where the beam gains of the eavesdropper are stronger than those of legitimate users with the purpose of increasing the secrecy sum-rate lower bound. In addition, for solving the power allocation problem, we developed an efficient iterative and convergent algorithm. The tightness of the secrecy sum-rate lower bound and the near-optimal performance of the proposed iterative algorithms which converge rapidly were demonstrated through simulations.
\appendices

\section{Proof of {\thref{th:eigenmatrix}}}\label{app:eigenmatrix_pf}

From the definition of the beam domain channel in \eqref{eq:def_beam_chan1} and \eqref{eq:def_beam_chan2}, the objective function of the problem in \eqref{prob:def_prob1} can be expressed as
\begin{align}\label{eq:sec_rate_re}
R(\tildebQ_{1},\cdots,\tildebQ_{K})=&\sum_{k=1}^{K}\bigg(\expect{\logdet{\bI+\bK_{k}^{-1}\tildebH_{k}\tildebQ_{k}\tildebH_{k}^{H}}}\notag\\
&-\log\det\Big(\bK_{\eve,k}\Big)\bigg)
\end{align}
where $\tildebQ_{k}=\bV^{H}\bQ_{k}\bV$, $k=1,\cdots,K$.

Utilizing the fact that random matrices $\tildebH_{k}$ ($\tildebH_{\eve}$) are column-independent with zero-mean entries, we have
\begin{align}
\bK_k&=\bI+\sum_{i\neq k}\expect{\bH_{k}\bQ_{i}\bH_{k}^{H}}\notag\\
&=\bI+\sum_{i\neq k}\expect{\tildebH_{k}\tildebQ_{i}\tildebH_{k}^{H}}\notag\\
&=\bI+\sum_{i\neq k}\expect{\tildebH_{k}\bXi_{i}\tildebH_{k}^{H}}\label{eq:K_k}
\end{align}
\begin{align}
\bK_{\eve,k}&=\bI+\expect{\bH_{\eve}\bQ_{k}\bH_{\eve}^{H}}\notag\\
&=\bI+\expect{\tildebH_{\eve}\tildebQ_{k}\tildebH_{\eve}^{H}}\notag\\
&=\bI+\expect{\tildebH_{\eve}\bXi_{k}\tildebH_{\eve}^{H}}\label{eq:K_eve}
\end{align}
with $\bXi_{k}=\tildebQ_{k}\odot\bI$ $(k=1,\cdots,K)$. Thus, the off-diagonal entries of $\tildebQ_{1},\cdots,\tildebQ_{K}$ do not affect the value of matrices $\bK_{k}$ and $\bK_{\eve,k}$.

Then, using a similar technique in \cite{Tulino06Capacity}, we define $\bPi_{m}\in\bbC^{M\times M}$ as a diagonal matrix whose diagonal entries are $1$ except the $(m,m)$th entry, which is $-1$. The entries of $\bPi_{m}\tildebQ_{k}\bPi_{m}$ are equivalent to those of $\tildebQ_{k}$ except the off-diagonals in the $m$th row and $m$th column, whose sign is reversed. Therefore, $\bPi_{m}\tildebQ_{k}\bPi_{m}\odot\bI=\tildebQ_{k}\odot\bI=\bXi_{k}$. Using equations \eqref{eq:K_k} and \eqref{eq:K_eve}, we have
\begin{align}
&R\left(\tildebQ_{1},\!\cdots\!,\tildebQ_{K}\right)=\sum_{k=1}^{K}\bigg(\expect{\logdet{\bI+\bK_k^{-1}\tildebH_{k}\tildebQ_{k}\tildebH_{k}^{H}}}\notag\\
&\qquad-\log\det\Big(\bK_{\eve,k}\Big)\bigg)\notag\\
&\equaa\sum_{k=1}^{K}\bigg(\expect{\logdet{\bI+\bK_k^{-1}\tildebH_{k}\bPi_{m}\tildebQ_{k}\bPi_{m}\tildebH_{k}^{H}}}\notag\\
&\qquad-\log\det\Big(\bK_{\eve,k}\Big)\bigg)\notag\\
&=R\left(\bPi_{m}\tildebQ_{1}\bPi_{m},\cdots,\bPi_{m}\tildebQ_{K}\bPi_{m}\right)
\end{align}
where (a) follows from the fact that reversing the sign of the $m$th column does not change its distribution. It is because that the columns of $\tildebH_{k}$ are independent and their distributions are symmetric.
Note the matrix $\half\left(\tildebQ_{k}+\bPi_{m}\tildebQ_{k}\bPi_{m}\right)$ has entries equal to those of $\tildebQ_{k}$ except for the off diagonal entries in the $m$th row and $m$th column, which are $0$. Moreover, due to the concavity of $\logdet{\cdot}$, we can invoke Jensen's inequality,
\begin{align}
&\quad R\bigg(\half\left(\tildebQ_{1}+\bPi_{m}\tildebQ_{1}\bPi_{m}\right),\!\cdots\!,\half\left(\tildebQ_{K}+\bPi_{m}\tildebQ_{K}\bPi_{m}\right)\bigg)\notag\\
&=\sum_{k=1}^{K}\!\Bigg(\!\expect{\logdet{\!\bI\!+\!\half\bK_k^{-1}\tildebH_{k}\!\left(\!\tildebQ_{k}\!+\!\bPi_{m}\tildebQ_{k}\bPi_{m}\right)\tildebH_{k}^{H}}\!}\notag\\
&\qquad\qquad-\log\det\Big(\bK_{\eve,k}\Big)\Bigg)\notag\\
&\geq\sum_{k=1}^{K}\Bigg(\half\expect{\logdet{\bI+\bK_k^{-1}\tildebH_{k}\tildebQ_{k}\tildebH_{k}^{H}}}\notag\\
&\qquad\qquad+\half\expect{\logdet{\bI+\bK_k^{-1}\tildebH_{k}\bPi_{m}\tildebQ_{k}\bPi_{m}\tildebH_{k}^{H}}}\notag\\
&\qquad\qquad-\log\det\Big(\bK_{\eve,k}\Big)\Bigg)\notag\\
&=\half R\left(\tildebQ_{1},\!\cdots\!,\tildebQ_{K}\right)
+\half R\left(\bPi_{m}\tildebQ_{1}\bPi_{m},\!\cdots\!,\bPi_{m}\tildebQ_{K}\bPi_{m}\right)\notag\\
&=R\left(\tildebQ_{1},\cdots,\tildebQ_{K}\right).
\end{align}
Note that nulling the off diagonal entries of any row and column of $\tildebQ_{1},\cdots,\tildebQ_{K}$ will increase the secrecy sum-rate lower bound in \eqref{eq:sec_rate_re}. Repeating this process for $m$ from $1$ to $M$, we find that \eqref{eq:sec_rate_re} is maximized when $\tildebQ_{1},\cdots,\tildebQ_{K}$ all are diagonal, which indicates that $\bPhi_{k}=\bV$. This concludes the proof.

\vspace{5mm}
\section{Proof of {\lmref{lm:ineq}}}\label{app:ineq_pf}

We need to prove that
\begin{equation}
\expect{\frac{x}{a+bx}}-\expect{\frac{\bar{x}}{a+bx}}\leq0
\end{equation}
where $a>0$, $b>0$, and $\bar{x}=\expect{x}$. By using Jensen's inequality, we have
\begin{align}
&\expect{\frac{x}{a+bx}}-\expect{\frac{\bar{x}}{a+bx}}\notag\\
=&\expect{\frac{x-\bar{x}}{a+bx}}\notag\\
\leq&\frac{\expect{x}-\bar{x}}{a+b\expect{x}}\notag\\
=&\frac{\bar{x}-\bar{x}}{a+b\expect{x}}\notag\\
=&0
\end{align}
where the  inequality follows from the fact that $\frac{x-\bar{x}}{a+bx}$ is a strictly concave function with regard to $x$ on $x\geq0$. The equality holds if and only if $x=\bar{x}$ with probability one for every $x$.

\section{Proof of {\thref{th:eigenvalue2}}}\label{app:eigenvalue2_pf}

For $N_r=1$, the optimization problem in \eqref{prob:def_prob2} can be rewritten as
\begin{align}\label{prob:def_prob_single}
&\argmax{\bLambda_{1},\cdots,\bLambda_{K}}\sum_{k=1}^{K}\Bigg(\expect{\log\bigg(1+\tilde{\bh}_{k}^{H}\bLambda_k\tilde{\bh}_{k}+\tr{\bLambda_{\backslash k}\tildebR_{k}}\bigg)}\notag\\
&\qquad-\log\left(1\!+\!\tr{\bLambda_{\backslash k}\tildebR_{k}}\right)
\!-\!\sum_{j=1}^{N_{e}}\log\left(1\!+\!\tr{\bLambda_k\check{\bR}_{\eve,j}}\right)\!\!\Bigg)\notag\\
&\quad\st\quad\sum_{k}\tr{\bLambda_{k}}\leq P\notag\\
&\quad\qquad\qquad\qquad\bLambda_{k}\succeq \mathbf{0},\quad k=1,\cdots,K
\end{align}
where $\bLambda_{\backslash k}=\sum_{k'\neq k}\bLambda_{k'}$, $\tilde{\bh}_{k}=[\tilde{h}_{k,1},\cdots,\tilde{h}_{k,M}]^{T}$, and $\check{\bR}_{\eve,j}=\diag{\bomega_{\eve,j}}$. Here, $\bomega_{\eve,j}$ is the $j$th row of matrix $\bOmega_{\eve}$.

Let diagonal matrices $\bPsi_k$ and $\mu$ be the Lagrange multipliers, and the cost function can expressed as
\begin{align}
\mathcal{L}&=
\sum_{k=1}^{K}\Bigg(\expect{\log\bigg(1+\tilde{\bh}_{k}^{H}\bLambda_k\tilde{\bh}_{k}+\tr{\bLambda_{\backslash k}\tildebR_{k}}\bigg)}\notag\\
&-\log\left(1\!+\!\tr{\bLambda_{\backslash k}\tildebR_{k}}\right)
-\sum_{j=1}^{N_{e}}\log\Big(1+\tr{\bLambda_k\check{\bR}_{\eve,j}}\!\Big)\!\Bigg)\notag\\
&+\sum_{k}\tr{\bPsi_k\bLambda_{k}}-\mu\left(\sum_{k}\tr{\bLambda_{k}}-P\right).
\end{align}

Then, let $\tilde{r}_{k,m}$, $\check{r}_{\eve,j,m}$, $\psi_{k,m}^{op}$, and $\lambda_{m}^{\op}$ be the $m$th diagonal entry of $\tildebR_{k}$, $\check{\bR}_{\eve,j}$, $\bPsi_{k}^{\op}$, and $\bLambda_{k}^{\op}$, respectively.
The Karush-Kuhn-Tucker (KKT) conditions for the optimal $\lambda_{k,m}^{\op}$, $\psi_{k,m}^{\op}$, and $\mu^{\op}$ can be expressed as
\begin{align}
&\pppd{\mathcal{L}}{\lambda_{k,m}^{\op}}=\expect{\frac{\tilde{h}_{k,m}\tilde{h}_{k,m}^{*}}{1+\tilde{\bh}_{k}^{H}\bLambda_{k}^{\op}\tilde{\bh}_{k}+\tr{\bLambda_{\backslash k}^{\op}\tildebR_{k}}}}\notag\\
&\!-\!\sum_{j=1}^{N_{e}}\frac{\check{r}_{\eve,j,m}}{1\!+\!\tr{\bLambda_{k}^{\op}\check{\bR}_{\eve,j}}}
-\sum_{k'\neq k}\left(\frac{\tilde{r}_{k',m}}{1+\tr{\bLambda_{\backslash k'}^{\op}\tildebR_{k'}}}\right.\notag\\
&\left.\!-\!\expect{\!\frac{\tilde{r}_{k',m}}{\!1\!+\!\tilde{\bh}_{k'}^{H}\bLambda_{k'}^{\op}\tilde{\bh}_{k'}\!+\!\tr{\bLambda_{\backslash k'}^{\op}\tildebR_{k'}}\!}\!}\!\right)
\!\!+\!\psi_{k,m}^{\op}\!-\!\mu^{\op}\!=\!0\!\!\label{eq:KKT_pppd}\\
&\psi_{k,m}^{\op}\lambda_{k,m}^{\op}\!=\!0,\ \mu^{op}\left(\sum_{k}\tr{\bLambda_{k}^{op}}\!-\!P\right)\!=\!0,\ \psi_{k,m}^{\op}\geq\!0\\
&\mu^{op}\geq0,\ \lambda_{k,m}^{op}\geq0,\ k=1,\cdots,K;\ m=1,\cdots,M.\label{eq:KKT_mu0}
\end{align}

Multiplying equation \eqref{eq:KKT_pppd} both sides by $\lambda_{k,m}^{op}$ and utilizing the condition $\psi_{k,m}^{\op}\lambda_{k,m}^{\op}=0$, we have
\begin{align}
&\mu^{op}\lambda_{k,m}^{op}=\expect{\frac{\lambda_{k,m}^{op}\tilde{h}_{k,m}\tilde{h}_{k,m}^{*}}{1+\tilde{\bh}_{k}^{H}\bLambda_{k}^{\op}\tilde{\bh}_{k}+\tr{\bLambda_{\backslash k}^{\op}\tildebR_{k}}}}\notag\\
&\qquad-\sum_{j=1}^{N_{e}}\frac{\lambda_{k,m}^{op}\check{r}_{\eve,j,m}}{1+\tr{\bLambda_{k}^{\op}\check{\bR}_{\eve,j}}}\notag\\
&\qquad-\sum_{k'\neq k}\left(\notag\frac{\lambda_{k,m}^{op}\tilde{r}_{k',m}}{1+\tr{\bLambda_{\backslash k'}^{\op}\tildebR_{k'}}}\right.\\
&\qquad-\left.\expect{\frac{\lambda_{k,m}^{op}\tilde{r}_{k',m}}{1+\tilde{\bh}_{k'}^{H}\bLambda_{k'}^{\op}\tilde{\bh}_{k'}+\tr{\bLambda_{\backslash k'}^{\op}\tildebR_{k'}}}}\right)\label{eq:KKT_mu}.
\end{align}
Repeat this process for $m$ from 1 to $M$ and summate, we can obtain \eqref{eq:KKT_sum} at the top of the next page.
\begin{figure*}[t]
\begin{align}
\mu^{op}\tr{\bLambda_{k}^{op}}&=\sum_{m=1}^M\left(\expect{\frac{\lambda_{k,m}^{op}\tilde{h}_{k,m}\tilde{h}_{k,m}^{*}}{1+\tilde{\bh}_{k}^{H}\bLambda_{k}^{\op}\tilde{\bh}_{k}+\tr{\bLambda_{\backslash k}^{\op}\tildebR_{k}}}}\right.
-\sum_{j=1}^{N_{e}}\frac{\lambda_{k,m}^{op}\check{r}_{\eve,j,m}}{1+\tr{\bLambda_{k}^{\op}\check{\bR}_{\eve,j}}}\notag\\
&\qquad-\sum_{k'\neq k}\left(\frac{\lambda_{k,m}^{op}\tilde{r}_{k',m}}{1+\tr{\bLambda_{\backslash k'}^{\op}\tildebR_{k'}}}\right.
\left.\left.-\expect{\frac{\lambda_{k,m}^{op}\tilde{r}_{k',m}}{1+\tilde{\bh}_{k'}^{H}\bLambda_{k'}^{\op}\tilde{\bh}_{k'}+\tr{\bLambda_{\backslash k'}^{\op}\tildebR_{k'}}}}\right)\right)\notag\\
&=1-\expect{\frac{1+\tr{\bLambda_{\backslash k}^{\op}\tildebR_{k}}}{1+\tilde{\bh}_{k}^{H}\bLambda_{k}^{\op}\tilde{\bh}_{k}+\tr{\bLambda_{\backslash k}^{\op}\tildebR_{k}}}}
-\sum_{j=1}^{N_{e}}\left(1-\frac{1}{1\!+\!\tr{\bLambda_{k}^{\op}\check{\bR}_{\eve,j}}}\!\right)\notag\\
&\qquad-\sum_{m=1}^M\!\sum_{k'\neq k}\!\left(\!\frac{\lambda_{k,m}^{op}\tilde{r}_{k',m}}{1\!+\!\tr{\bLambda_{\backslash k'}^{\op}\tildebR_{k'}}}\right.
-\left.\expect{\frac{\lambda_{k,m}^{op}\tilde{r}_{k',m}}{1+\tilde{\bh}_{k'}^{H}\bLambda_{k'}^{\op}\tilde{\bh}_{k'}+\tr{\bLambda_{\backslash k'}^{\op}\tildebR_{k'}}}}\right)\label{eq:KKT_sum}
\end{align}
\hrulefill
\end{figure*}

Due to the conditions $\mu^{op}\geq0$ and $\lambda_{k,m}^{op}\geq0$, it is east to find that
\begin{equation}
\mu^{op}\tr{\bLambda_{k}^{op}}\geq0.
\end{equation}
Furthermore, observing the equation \eqref{eq:KKT_sum}, we have
\begin{align}
&\sum_{m=1}^M\sum_{k'\neq k}\left(\frac{\lambda_{k,m}^{op}\tilde{r}_{k',m}}{1+\tr{\bLambda_{\backslash k'}^{\op}\tildebR_{k'}}}\right.\notag\\
&\qquad\left.-\expect{\frac{\lambda_{k,m}^{op}\tilde{r}_{k',m}}{1+\tilde{\bh}_{k'}^{H}\bLambda_{k'}^{\op}\tilde{\bh}_{k'}+\tr{\bLambda_{\backslash k'}^{\op}\tildebR_{k'}}}}\right)\geq0.\label{eq:KKT_leq}
\end{align}
Thus, combining \eqref{eq:KKT_sum}-\eqref{eq:KKT_leq}, we can obtain the following inequalities
\begin{align}
&1-\expect{\frac{1+\tr{\bLambda_{\backslash k}^{\op}\tildebR_{k}}}{1+\tilde{\bh}_{k}^{H}\bLambda_{k}^{\op}\tilde{\bh}_{k}+\tr{\bLambda_{\backslash k}^{\op}\tildebR_{k}}}}\notag\\
&\qquad\qquad-\maximize{j=1,\cdots,N_e}\left(1-\frac{1}{1+\tr{\bLambda_{k}^{\op}\check{\bR}_{\eve,j}}}\right)\notag\\
\geq&1-\expect{\frac{1+\tr{\bLambda_{\backslash k}^{\op}\tildebR_{k}}}{1+\tilde{\bh}_{k}^{H}\bLambda_{k}^{\op}\tilde{\bh}_{k}+\tr{\bLambda_{\backslash k}^{\op}\tildebR_{k}}}}\notag\\
&\qquad\qquad-\sum_{j=1}^{N_{e}}\left(1-\frac{1}{1+\tr{\bLambda_{k}^{\op}\check{\bR}_{\eve,j}}}\right)\notag\\
\geq&0.
\end{align}
Here, without loss of generality, we suppose that
\begin{align}
j_{\mm} = \argmax{j=1,\cdots,N_e}\left(1-\frac{1}{1+\tr{\bLambda_{k}^{\op}\check{\bR}_{\eve,j}}}\right).
\end{align}
Consequently,
\begin{align}
\!\!\!\!\frac{1}{1\!+\!\tr{\bLambda_{k}^{\op}\check{\bR}_{\eve,j_{\mm}}}}&\!\!\geq\!
\expect{\!\!\frac{1}{1\!+\!\tilde{\bh}_{k}^{H}\bLambda_{k}^{\op}\tilde{\bh}_{k}\!+\!\tr{\bLambda_{\backslash k}^{\op}\tildebR_{k}}}\!\!}\label{eq:KKT_re}\\
\tr{\bLambda_{k}^{\op}\check{\bR}_{\eve,j_{\mm}}}&\!\!\geq\tr{\bLambda_{k}^{\op}\check{\bR}_{\eve,j}},\ j=1,\cdots,N_e.\label{eq:KKT_re1}
\end{align}

Notice that for the $m$th beam of the $k$th user, we have $[\tildebR_{\eve}]_{mm}\ge[\tildebR_{k}]_{mm}$, which indicates that $\tilde{r}_{\eve,m}\geq\tilde{r}_{k,m}$ with $\sum_{j=1}^{N_{e}}\check{r}_{\eve,j,m}=\tilde{r}_{\eve,m}$. Then, if $\lambda_{k,m}^{\op}\neq0$, from the equation \eqref{eq:KKT_mu}, we can obtain \eqref{eq:ineq_mu} at the top of the this page.
where (a) follows from the stochastic independence between the elements of $\tilde{\bh}_{k}$ and Lemma 1, while (b) can be directly obtained from \eqref{eq:KKT_re}-\eqref{eq:KKT_re1} and $\sum_{j=1}^{N_{e}}\check{r}_{\eve,j,m}=\tilde{r}_{\eve,m}$. Owing to $\mu^{op}<0$ does not satisfy the KKT condition in \eqref{eq:KKT_mu0}, consequently, $\lambda_{k,m}^{\op}=0$. This completes the proof.

\begin{figure*}[t]
\begin{align}\label{eq:ineq_mu}
\mu^{op}&=\expect{\frac{\tilde{h}_{k,m}\tilde{h}_{k,m}^{*}}{1+\tilde{\bh}_{k}^{H}\bLambda_{k}^{\op}\tilde{\bh}_{k}+\tr{\bLambda_{\backslash k}^{\op}\tildebR_{k}}}}
-\sum_{j=1}^{N_{e}}\frac{\check{r}_{\eve,j,m}}{1+\tr{\bLambda_{k}^{\op}\check{\bR}_{\eve,j}}}\notag\\
&\qquad\qquad-\sum_{k'\neq k}\left(\frac{\tilde{r}_{k',m}}{1+\tr{\bLambda_{\backslash k'}^{\op}\tildebR_{k'}}}\right.
\left.-\expect{\frac{\tilde{r}_{k',m}}{1+\tilde{\bh}_{k'}^{H}\bLambda_{k'}^{\op}\tilde{\bh}_{k'}+\tr{\bLambda_{\backslash k'}^{\op}\tildebR_{k'}}}}\right)\notag\\
&\leq\expect{\frac{\tilde{h}_{k,m}\tilde{h}_{k,m}^{*}}{1+\tilde{\bh}_{k}^{H}\bLambda_{k}^{\op}\tilde{\bh}_{k}+\tr{\bLambda_{\backslash k}^{\op}\tildebR_{k}}}}
-\sum_{j=1}^{N_{e}}\frac{\check{r}_{\eve,j,m}}{1+\tr{\bLambda_{k}^{\op}\check{\bR}_{\eve,j}}}\notag\\
&\mathop{<}^{(\textrm{a})}\expect{\frac{\tilde{r}_{k,m}}{1+\tilde{\bh}_{k}^{H}\bLambda_{k}^{\op}\tilde{\bh}_{k}+\tr{\bLambda_{\backslash k}^{\op}\tildebR_{k}}}}
-\sum_{j=1}^{N_{e}}\frac{\check{r}_{\eve,j,m}}{1+\tr{\bLambda_{k}^{\op}\check{\bR}_{\eve,j}}}\notag\\
&\mathop{\leq}^{(\textrm{b})}\frac{\tilde{r}_{k,m}}{1+\tr{\bLambda_{k}^{\op}\check{\bR}_{\eve,j_{\mm}}}}
-\frac{\tilde{r}_{\eve,m}}{1+\tr{\bLambda_{k}^{\op}\check{\bR}_{\eve,j_{\mm}}}}\notag\\
&\leq0.
\end{align}
\hrulefill
\end{figure*}

\begin{figure*}[!t]
\begin{align*}\label{eq:det_KKT_final}
\tag{80}
\begin{cases}
\frac{\gamma_{k,m}^{(i+1)}}{1+\gamma_{k,m}^{(i+1)}\lambda_{k,m}^{(i+1)}}+\sum\limits_{l \neq k}^{K}\sum\limits_{j=1}^{N_{r}}\frac{\check{r}_{l,m,j}}{\tilde{\gamma}_{l,j}^{(i+1)}+\tr{\check{\bR}_{l,j}\bLambda_{\backslash l}^{(i+1)}}}=\delta_{k,m}^{(i)}+\mu^{(i+1)},\quad\quad&\mu^{(i+1)}<\nu_{k,m}^{(i+1)}-\delta_{k,m}^{(i)}\\
\lambda_{k,m}^{(i+1)}=0,\quad\quad&\mu^{(i+1)}\geq\nu_{k,m}^{(i+1)}-\delta_{k,m}^{(i)}
\end{cases}
\end{align*}
\hrulefill
\end{figure*}

\begin{figure*}[!t]
\begin{align*}\label{eq:C2}
\tag{84}
\bar{C}^{(i)}(\bar{\bX}_{1},\cdots,\bar{\bX}_{K})\leq\tilde{C}^{(i)}\left(\frac{1}{KM}\bar{\bX}_{1}+\frac{KM-1}{KM}\bX_{1}^{(t)},\cdots,\frac{1}{KM}\bar{\bX}_{K}+\frac{KM-1}{KM}\bX_{K}^{(t)}\right)
\end{align*}
\hrulefill
\end{figure*}

\begin{figure*}[!t]
\begin{align*}\label{eq:IWFA_conv}
\tag{85}
\tilde{C}^{(i)}(\bX_{1}^{(t)},\cdots,\bX_{K}^{(t)})\leq\tilde{C}^{(i)}\left(\frac{1}{KM}\bar{\bX}_{1}+\frac{KM-1}{KM}\bX_{1}^{(t)},\cdots,\frac{1}{KM}\bar{\bX}_{K}+\frac{KM-1}{KM}\bX_{K}^{(t)}\right)
\end{align*}
\hrulefill
\end{figure*}

\vspace{5mm}
\section{Proof of {\thref{th:deterministic}}}\label{app:deterministic_pf}

The cost function of the problem in \eqref{prob:def_prob5} can be defined as
\begin{align}
\mathcal{L}=&\sum_{k}\overline{\cR}_{k,1}(\bLambda_{1},\cdots,\bLambda_{K})\notag\\
&\quad-\sum_{k}\tr{\left(\ppd{\bLambda_{k}}\sum_{l}\cR_{l,2}(\bLambda_{1}^{(i)},\cdots,\bLambda_{K}^{(i)})\right)^{T}\bLambda_{k}}\notag\\
&\quad+\sum_{k=1}^{K}\tr{\bPsi_{k}\bLambda_{k}}-\mu\left(\tr{\sum_{k}\bLambda_{k}}-P\right)
\end{align}
where the Lagrange multipliers $\bPsi_{k}\succeq\bzero$ and $\mu\geq0$ are related to the problem constraints.

From \eqref{eq:DE_R1}, the gradient of $\overline{\cR}_{k,1}(\bLambda_{1},\cdots,\bLambda_{K})$ with respect to $\bLambda_{k}$ can be expressed as
\begin{align}
&\!\!\!\ppd{\bLambda_{k}}\overline{\cR}_{k,1}(\bLambda_{1},\cdots,\bLambda_{K})=\left(\bI+\bGamma_{k}\bLambda_{k}\right)^{-1}\bGamma_{k}\notag\\
&+\sum_{m,n}\pppd{\overline{\cR}_{k,1}(\bLambda_{1},\cdots,\bLambda_{K})}{\left[\tilde{\eta}_{k}\left(\bPhi_{k}^{-1}\bLambda_{k}\right)\right]_{mn}}
\pppd{\left[\tilde{\eta}_{k}\left(\bPhi_{k}^{-1}\bLambda_{k}\right)\right]_{mn}}{\bLambda_{k}}\notag\\
&+\sum_{m,n}\pppd{\overline{\cR}_{k,1}(\bLambda_{1},\cdots,\bLambda_{K})}{\left[\eta_{k}\left(\tilde{\bPhi}_{k}^{-1}\barbK_{k}^{-1}\right)\right]_{mn}}
\pppd{\left[\eta_{k}\left(\tilde{\bPhi}_{k}^{-1}\barbK_{k}^{-1}\right)\right]_{mn}}{\bLambda_{k}}.
\end{align}
Using methods similar to that in the proof of Theorem 4 in \cite{Lu16Deterministic}, we obtain
\begin{align}
\pppd{\overline{\cR}_{k,1}(\bLambda_{1},\cdots,\bLambda_{K})}{\left[\tilde{\eta}_{k}\left(\bPhi_{k}^{-1}\bLambda_{k}\right)\right]_{mn}}&=0\\
\pppd{\overline{\cR}_{k,1}(\bLambda_{1},\cdots,\bLambda_{K})}{\left[\eta_{k}\left(\tilde{\bPhi}_{k}^{-1}\barbK_{k}^{-1}\right)\right]_{mn}}&=0.
\end{align}
Thus, we have
\begin{equation}\label{eq:der1}
\ppd{\bLambda_{k}}\overline{\cR}_{k,1}(\bLambda_{1},\cdots,\bLambda_{K})=\left(\bI+\bGamma_{k}\bLambda_{k}\right)^{-1}\bGamma_{k}.
\end{equation}

Then, similarly to process of obtaining the gradient of $\overline{\cR}_{k,1}(\bLambda_{1},\cdots,\bLambda_{K})$ with respect to $\bLambda_{k}$, we can obtain the gradient of $\overline{\cR}_{k,1}(\bLambda_{1},\cdots,\bLambda_{K})$ with respect to $\bLambda_{k'}$,\ $k'\neq k$. We have
\begin{equation}\label{eq:der2}
\ppd{\bLambda_{k'}}\overline{\cR}_{k,1}(\bLambda_{1},\!\cdots\!,\bLambda_{K})=\sum_{j=1}^{N_{r}}\frac{\check{\bR}_{k',j}}{\tilde{\gamma}_{k',j}+\tr{\bLambda_{\backslash k'}\check{\bR}_{k',j}}}.
\end{equation}
Therefore, from \eqref{eq:der1}-\eqref{eq:der2}, we have
\begin{align}
\ppd{\bLambda_{a}}\sum_{k}\overline{\cR}_{k,1}&(\bLambda_{1},\cdots,\bLambda_{K})=\left(\bI+\bGamma_{a}\bLambda_{a}\right)^{-1}\bGamma_{a}\notag\\
&\qquad+\sum_{l\neq a}^{K}\sum_{j=1}^{N_{r}}\frac{\check{\bR}_{l,j}}{\tilde{\gamma}_{l,j}+\tr{\bLambda_{\backslash l}\check{\bR}_{l,j}}}.
\end{align}

%
%

Moreover, the derivative of $\sum_{l}\cR_{l,2}(\bLambda_{1},\cdots,\bLambda_{K})$ with respect to $\bLambda_{k}$ can be expressed as
\begin{align}\label{eq:det_der}
\ppd{\bLambda_{k}}\sum_{l}\cR_{l,2}&(\bLambda_{1},\cdots,\bLambda_{K})=\sum_{l\neq k}^{K}\sum_{j=1}^{N_{r}}\frac{\check{\bR}_{l,j}}{1+\tr{\bLambda_{\backslash l}\check{\bR}_{l,j}}}\notag\\
&\qquad\qquad+\sum_{j=1}^{N_{e}}\frac{\check{\bR}_{\eve,j}}{1+\tr{\bLambda_{k}\check{\bR}_{\eve,j}}}.
\end{align}
Due to the fact that $\overline{\cR}_{k,1}(\bLambda_{1},\cdots,\bLambda_{K})$ is strictly concave on $(\bLambda_{1},\cdots,\bLambda_{K})$, the KKT conditions of \eqref{prob:def_prob5} are
\begin{align}
&\pppd{\mathcal{L}}{\bLambda_{a}^{(i+1)}}=\bzero,\ a=1,\cdots,K\label{eq:det_KKT1}\\
&\tr{\bPsi_{a}^{(i+1)}\bLambda_{a}^{(i+1)}}=0,\ \bPsi_{a}^{(i+1)}\succeq\bzero,\ \bLambda_{a}^{(i+1)}\succeq\bzero\\
&\mu^{(i+1)}\left(\tr{\sum_{k}\bLambda_{k}^{(i+1)}}-P\right)=0,\ \mu^{(i+1)}\geq0.\label{eq:det_KKT3}
\end{align}

Since problem in \eqref{prob:def_prob5} is a convex optimization problem, the optimal solution $\bLambda_{k}^{(i+1)}$ can be determined by solving the KKT conditions. We rewrite the first KKT condition \eqref{eq:det_KKT1} as
\begin{align}
&\pppd{\mathcal{L}}{\bLambda_{a}^{(i+1)}}=\left(\bI+\bGamma_{a}^{(i+1)}\bLambda_{a}^{(i+1)}\right)^{-1}\bGamma_{a}^{(i+1)}\notag\\
&\!+\!\sum_{l\neq a}^{K}\!\sum_{j=1}^{N_{r}}\!\frac{\check{\bR}_{l,j}}{\tilde{\gamma}_{l,j}^{(i+1)}\!+\!\tr{\bLambda_{\backslash l}^{(i+1)}\check{\bR}_{l,j}}}
\!-\!\sum_{l\neq a}^{K}\!\sum_{j=1}^{N_{r}}\!\frac{\check{\bR}_{l,j}}{1\!+\!\tr{\bLambda_{\backslash l}^{(i)}\check{\bR}_{l,j}}}\notag\\
&\!-\!\sum_{j=1}^{N_{e}}\frac{\check{\bR}_{\eve,j}}{1+\tr{\bLambda_{a}^{(i)}\check{\bR}_{\eve,j}}}+\bPsi_{a}^{(i+1)}-\mu^{(i+1)}\bI=\bzero.
\end{align}

Then, similarly to \cite{Lu16Deterministic,Zhang13MAC}, It can be found that the KKT conditions \eqref{eq:det_KKT1}-\eqref{eq:det_KKT3} equal to the KKT conditions of the following optimization problem
\begin{align}\label{prob:def_prob7}
&\left[\bLambda_{1}^{(i+1)},\cdots,\bLambda_{K}^{(i+1)}\right]=\argmax{\bLambda_{1},\cdots,\bLambda_{K}}\sum_{k}
\bigg(\logdet{\bI+\bGamma_{k}\bLambda_{k}}\notag\\
&\qquad+\logdet{\tilde{\bGamma}_{k}+\barbK_{k}}-\tr{\bDelta_{k}^{(i)}\bLambda_{k}}\bigg)\notag\\
&\qquad\qquad\st\quad\sum_{k}\tr{\bLambda_{k}}\leq P\notag\\
&\qquad\qquad\qquad\qquad\qquad\bLambda_{k}\succeq \mathbf{0},\quad k=1,\cdots,K.
\end{align}


Thus, the solution of the iterative problem in \eqref{prob:def_prob5} is equivalent to that of problem in \eqref{prob:def_prob7}. Using the KKT conditions of problem in \eqref{prob:def_prob7}, we can find that the optimal solution of \eqref{prob:def_prob7} satisfies \eqref{eq:det_KKT_final}, which is given at the top of this page.
In \eqref{eq:det_KKT_final}, the auxiliary variable $\nu_{k,m}^{(i+1)}$ is given by
\begin{equation}
\setcounter{equation}{81}
\nu_{k,m}^{(i+1)}=\gamma_{k,m}^{(i+1)}+\sum_{l \neq k}^{K}\!\sum_{j=1}^{N_{r}}\frac{\check{r}_{l,m,j}}{\tilde{\gamma}_{l,j}^{(i+1)}
\!+\!\sum\limits_{\substack{(l',m')\\\in\mathcal{S}_{k,m,l}}}\check{r}_{l,m',j}\lambda_{l',m'}^{(i+1)}}
\end{equation}
and the Lagrange multiplier $\mu^{(i+1)}$ is set to satisfy the KKT conditions $\mu^{(i+1)}\left(\tr{\sum_{k}\bLambda_{k}^{(i+1)}}-P\right)=0$ and $\mu^{(i+1)}\geq0$.

\section{Proof of {\thref{th:alg}}}\label{app:alg_pf}

Here, we prove the convergence of the sequence $\big\{\bX_{1}^{(t)},\!\cdots\!,\bX_{K}^{(t)}\big\}_{t=0}^{\infty}$.
To prove the convergence of the proposed IWFA, a function for given $(\bX_{1}^{(t)},\!\cdots\!,\bX_{K}^{(t)})$ is defined as
\begin{align}
&\bar{C}^{(i)}(\bX_{1},\cdots,\bX_{K})=\frac{1}{KM}\sum_{k}^{K}\sum_{m}^{M}
\Bigg(\log\left(1+\gamma_{k,m}^{(i)}x_{k,m}\right)\notag\\
&+\sum_{k'\neq k}\sum_{j}^{N_{r}}\log\bigg(\tilde{\gamma}_{k',j}^{(i)}+\check{r}_{k',m,j}x_{k,m}
+\!\!\sum\limits_{\substack{(l,m')\\\in\mathcal{S}_{k,m,k'}}}\!\!\check{r}_{k',m',j}x_{l,m'}^{(t)}\bigg)\notag\\
&-\delta_{k,m}^{(i)}x_{k,m}
+\sum_{\substack{(k',m')\\\neq (k,m)}}\log\left(1+\gamma_{k',m'}^{(i)}x_{k',m'}^{(t)}\right)\notag\\
&+\!\sum_{j}^{N_{r}}\log\!\bigg(\!\tilde{\gamma}_{k,j}^{(i)}\!+\!\sum_{l\neq k}\sum_{m'}\check{r}_{k,m',j}x_{l,m'}^{(t)}\!\!\bigg)
\!\!-\!\!\!\!\sum_{\substack{(l,m')\\\neq(k,m)}}\!\!\!\delta_{l,m'}^{(i)}x_{l,m'}^{(t)}\!\!\Bigg).\!\!
\end{align}
It can be seen that $\bar{C}^{(i)}(\bX_{1},\cdots,\bX_{K})$ is a concave function with respect to $(\bX_{1},\cdots,\bX_{K})$. For given $(\bX_{1}^{(t)},\cdots,\bX_{K}^{(t)})$ subject to the power constraint $\sum_{k,m}x_{k,m}\leq P$, the $(\bar{\bX}_{1},\cdots,\bar{\bX}_{K})$ generated in Step 14 of IWFA is exactly equal to the solution of maximizing $\bar{C}^{(i)}(\bX_{1},\cdots,\bX_{K})$. Hence, with the $(\bar{\bX}_{1},\cdots,\bar{\bX}_{K})$ obtained from Step 14 of the IWFA, we have the following result
\begin{align}\label{eq:C1}
\bar{C}^{(i)}(\bar{\bX}_{1},\cdots,\bar{\bX}_{K})&\geq\bar{C}^{(i)}(\bX_{1}^{(t)},\cdots,\bX_{K}^{(t)})\notag\\
&=\tilde{C}^{(i)}(\bX_{1}^{(t)},\cdots,\bX_{K}^{(t)}).
\end{align}
From the concavity of $\tilde{C}^{(i)}(\bX_{1},\cdots,\bX_{K})$, it can be shown that \eqref{eq:C2} holds.
Then, combing \eqref{eq:C1} and \eqref{eq:C2}, we have the inequality \eqref{eq:IWFA_conv}. Note that both \eqref{eq:C2} and \eqref{eq:IWFA_conv} are given at the top of the previous page.

Therefore, after Step 20 of the IWFA, we have that $\tilde{C}^{(i)}(\bX_{1}^{(t+1)},\!\cdots\!,\bX_{K}^{(t+1)})\!\geq\!\tilde{C}^{(i)}(\bX_{1}^{(t)},\!\cdots\!,\bX_{K}^{(t)})$. Furthermore, the problem in \eqref{prob:def_prob6} is a convex problem. Thus, the generated $(\bX_{1}^{(t+1)},\!\cdots\!,\bX_{K}^{(t+1)})$ will be convergent.

\bibliographystyle{IEEEtran}
\bibliography{Refabrv_20161223,References_20161129}
\end{document}